\documentclass[fontsize=11pt,toc=bib]{scrartcl}

\setkomafont{title}{\sffamily}
\setkomafont{pagehead}{\sffamily}
\setkomafont{pagenumber}{\sffamily}
\setkomafont{author}{\sffamily}
\addtokomafont{author}{\vspace*{2em}\large}

\setkomafont{dedication}{\footnotesize\raggedleft\itshape}
\usepackage[automark]{scrlayer-scrpage}
\ihead{\headmark}
\ohead{\pagemark}
\usepackage[
  paper=a4paper,
  left=80pt,
  right=80pt,
  top=86pt,
  bottom=56pt,
  bindingoffset=0mm,
  includefoot,
  includehead,
  foot=16.5pt,
  head=16.5pt
]{geometry}

\DeclareMathAlphabet{\pazocal}{OMS}{zplm}{m}{n}
\DeclareMathAlphabet{\mymathbb}{U}{bbold}{m}{n}
\newcommand{\One}{\mymathbb{1}}
\newcommand{\Zero}{\mymathbb{0}}

\usepackage{comment}
\usepackage{enumitem}
\usepackage{mathtools}
\usepackage{fontawesome}
\usepackage{booktabs}
\usepackage{indentfirst}
\usepackage{lipsum}
\usepackage[english]{babel}
\usepackage{amsmath} 
\usepackage{amsthm}
\usepackage[sfdefault]{FiraSans}
\usepackage[charter,expert]{mathdesign}
\usepackage[scaled=.96]{XCharter}
\usepackage{setspace}
\usepackage{nicefrac}
\usepackage[dvipsnames]{xcolor}
    \definecolor{linkcolor}{HTML}{0D6A9E}
    \definecolor{3blue}{HTML}{0072B2}
    \definecolor{3green}{HTML}{009E73}
    \definecolor{3ochre}{HTML}{E69F00}
    \definecolor{3yellow}{HTML}{F0E442}
    \definecolor{3cyan}{HTML}{56B4E9}
    \definecolor{3red}{HTML}{D55E00}
    \definecolor{3pink}{HTML}{CC79A7}
    \definecolor{2blue}{HTML}{1A85FF}
    \definecolor{2red}{HTML}{D41159}
    \definecolor{mygray}{HTML}{989898}
\usepackage{graphicx}
    \graphicspath{ {./gfx/} }
\usepackage{caption}
\usepackage{tensor}
\usepackage{tikz-cd}
    \tikzcdset{arrow style=tikz}
    \tikzset{vertl/.style={anchor=south, rotate=90, inner sep=1mm}}
    \tikzset{vertr/.style={anchor=south, rotate=-90, inner sep=1mm}}
\usepackage[hang,flushmargin]{footmisc} 
\usepackage{hyperref}
    \hypersetup{
    colorlinks=true,
    linktoc=all,
    linkcolor=linkcolor,
    citecolor=linkcolor,
    filecolor=magenta,      
    urlcolor=linkcolor,
    pdfpagemode=FullScreen,
    pdfauthor=Davide Dal Martello,
    pdftitle=Okamoto's symmetry on the representation space of the sixth Painlevé equation
    }
\usepackage{cleveref}
    \newtheoremstyle{komait}
      {\topsep}   
      {\topsep}   
      {\itshape}  
      {0pt}       
      {\bfseries\sffamily} 
      {.}         
      {5pt plus 1pt minus 1pt} 
      {}          
    \newtheoremstyle{komanormal}
      {\topsep}   
      {\topsep}   
      {\rmfamily}  
      {0pt}       
      {\bfseries\sffamily} 
      {.}         
      {5pt plus 1pt minus 1pt} 
      {}          
    \theoremstyle{komait}
    \newtheorem{theorem}{Theorem}[section]

    \newtheorem{definition}[theorem]{Definition}

    \newtheorem{problem}{Problem}
    \theoremstyle{komanormal}
    \NewCommandCopy{\proofqedsymbol}{\qedsymbol}
        \AtBeginEnvironment{proof}{\renewcommand{\qedsymbol}{\proofqedsymbol}}
    \newtheorem{example}[theorem]{Example}
        \AtBeginEnvironment{example}{\pushQED{\qed}\renewcommand{\qedsymbol}{$\triangle$}}
        \AtEndEnvironment{example}{\popQED\endexample}
    
        \AtBeginEnvironment{exercise}{\pushQED{\qed}\renewcommand{\qedsymbol}{$\triangle$}}
        \AtEndEnvironment{exercise}{\popQED\endexample}
    \newtheorem{remark}[theorem]{Remark}
        \AtBeginEnvironment{remark}{\pushQED{\qed}\renewcommand{\qedsymbol}{$\triangle$}}
        \AtEndEnvironment{remark}{\popQED\endexample}
    
\newcommand{\MC}{M\!C}

\begin{document}

\title{\usekomafont{subtitle}\LARGE\vspace{-5em}Okamoto's symmetry on the representation space\\of the sixth Painlevé equation}
\author{\raisebox{-.5ex}{\href{https://orcid.org/0000-0002-4975-8774}{\includegraphics[height=15pt]{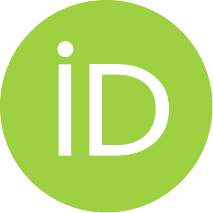}}}\hspace{.5em}Davide DAL MARTELLO\footnote{\hspace{.4em}Department of Mathematics ``Tullio Levi-Civita'', Università degli studi di Padova, Via Trieste 63, 35121 Padova (Italy)\\
\faEnvelopeO\hspace*{.5em}davide.dalmartello@unipd.it, \href{mailto:contact@davidedalmartello.com}{contact@davidedalmartello.com}}}
\dedication{%
  In loving memory of Masatoshi Noumi.\\
  May this be a stepping stone in fulfilling\\
  the last inspiring wish you had for me.\vspace{-2.5em}
}
\date{\vspace{-1em}}

\maketitle

\begin{abstract}\noindent
The sixth Painlevé equation ($P_{V\!I}$) admits dual isomonodromy representations of type $2$-dimensional Fuchsian and $3$-dimensional Birkhoff. 
Taking the multiplicative middle convolution of a Teichm\"uller $\pazocal{X}$-coordinatization for the Fuchsian monodromy group, we give Okamoto's symmetry $w_2$ of $P_{V\!I}$ a monodromic realization in the language of cluster $\pazocal{X}$-mutations. The explicit mutation formula is encoded in dual geometric terms of colored equilateral triangulations and star-shaped fat graphs. Moreover, this realization has a known additive analogue through the middle convolution for Fuchsian systems, and dual formulations for both the Birkhoff representation and its Stokes data exist. We give this quadruple of maps, each one realizing $w_2$, a unified diagrammatic description in purely convolutional terms.
\end{abstract}

\noindent
\begin{center}{\vspace{-.5em}\small{\textbf{\textsf{Keywords}}\hspace{.5em}Painlevé VI, middle convolution, cluster algebra.}}
\end{center}

\renewcommand{\contentsname}{\textcolor{linkcolor}{Contents}}

\tableofcontents

\section{Introduction}

The \textbf{sixth Painlevé equation} ($P_{V\!I}$), first encountered by Paul Painlevé \cite{Painleve1900} while in a search for new special functions, is nowadays an integral part of the mathematical physicist's toolkit.
Despite the naming, the full form of $P_{V\!I}$ is due to Richard Fuchs's study \cite{Fuchs1905} of monodromy preserving deformations of the eponymous type of linear system 
\begin{equation*}\label{Fuchn}
\frac{\mathrm d}{{\mathrm d}\lambda} \Phi=\left(
\sum_{k=1}^{p}\frac{{A}_k}{\lambda-u_k} \right)\Phi,\quad A_k\in\mathfrak{g},
\end{equation*}
for a Lie algebra $\mathfrak{g}$ and pairwise distinct complex constants $u_1,\ldots,u_p$.
These deformations are characterized by the Schlesinger equations, whose specialization to $\mathfrak{sl}_2(\mathbb{C})$-systems with finite singularities $(u_1,u_2,u_3)=(0,1,t)$ and spectral data $\{\theta_1,\theta_2,\theta_3,\theta_\infty\}$ indeed reduces to $P_{V\!I}(\theta)$.

Okamoto \cite{Okamoto1987} showed that the solution space, whose distinguished special functions are known as the $P_{V\!I}$ transcendents, admits a group of symmetries given by B\"acklund birational transformations.
Each symmetry maps solutions to solutions by changing parameters as an element of $W(\tilde{D}_4)$, the affine Weyl group of type $D_4$.
Inaba, Iwasaki, and Saito \cite{Inaba2004} later elucidated the true genesis of this group by switching to the monodromic viewpoint: Okamoto's B\"acklund transformations on the de Rham moduli space are singled out by being those covering the identity on the Betti moduli space through the Riemann-Hilbert (RH) correspondence.
More precisely, let
\begin{equation}\label{Aspace}
    \pazocal{A}(\theta)=\pazocal{F}(\theta)/\Gamma,
\end{equation}
for $\Gamma$ the gauge group and the \textbf{connection space}
\begin{equation}\label{Sysspace}
\begin{aligned}
    \pazocal{F}(\theta)=\bigg\{ {\mathrm d}-\left(\frac{{A}_1}{\lambda}+\frac{{A}_2}{\lambda-1}+\frac{{A}_3}{\lambda-t}\right){{\mathrm d}\lambda} \,\Big|\, &{A}_1+{A}_2+{A}_3+{A}_\infty=\Zero,\\[-.5em]&\mathrm{eig}({A}_k)=\left\{\pm\textstyle\frac{\theta_k}{2}\right\} \text{ for } k=1,2,3,\infty\bigg\},
\end{aligned}
\end{equation}
denote the moduli space of meromorphic $\mathfrak{sl}_2(\mathbb{C})$-connections on the four-punctured Riemann sphere $\Sigma_{0,4}$ and
\begin{equation}
\pazocal{M}(a)=\pazocal{R}(\iota)/SL_2(\mathbb{C}),
\end{equation}
for the \textbf{representation space}
\begin{equation}\label{Repspace}
\begin{aligned}
    \pazocal{R}(\iota)=\Big\{ (M_1,M_2,M_3,M_\infty) \,\big|\, &M_1M_2M_3M_\infty=\One,\\[-.5em]&\hbox{eig}(M_k)= \big\{\iota_k^{\pm1}\big\} \text{ for } k=1,2,3,\infty\Big\}
\end{aligned}
\end{equation}
and local data $a_k=\mathrm{Tr}(M_k)$,
be the associated moduli space of monodromy $SL_2(\mathbb{C})$-representations.

\begin{remark}
    Let us point out that $\iota_k=e^{\pi i \theta_k}$: multi-valuedness of a fundamental solution $\Phi(\lambda)$ is encapsulated by the $SL_2(\mathbb C)$-subgroup of monodromy matrices
\begin{equation}\label{Fuchdata}
\left\langle M_1,M_2,M_3,M_\infty \,\big|\, M_1M_2M_3M_\infty=\One\right\rangle
\end{equation}
whose generators, essentially, exponentiate the matrix residua of connections \eqref{Sysspace}.
\end{remark}

Recombining local data as
\begin{equation}\begin{cases}\label{omega}
    \omega_i=a_ia_4+a_ja_k, \quad i=1,2,3, \\
    \omega_4=a_1a_2a_3a_4+a_1^2+a_2^2+a_3^2+a_4^2-4,
\end{cases}\end{equation}
the latter moduli space is well-known to be identified via the cyclic assignment $x_i=\mathrm{tr}(M_jM_k)$ with the \textbf{monodromy manifold} $\pazocal{S}(\omega)$, i.e., the Fricke cubic surface
\begin{equation}\label{mon_mfd}
    \pazocal{S}(\omega)=\big\{(x_1,x_2,x_3)\in\mathbb{C}^3 \,\big|\,x_1x_2x_3+x_1^2+x_2^2+x_3^2-\omega_1x_1-\omega_2x_2-\omega_3x_3-\omega_4=0\big\}.
\end{equation}
 Then, for $w_i\in W(\tilde{D}_4)$ a generating reflection---realized as a change of parameters $\theta$---whose lift as B\"acklund transformation we denote by $s_i$, the following square
\begin{equation}\label{Masahiko}
        \begin{tikzcd}[row sep=large, column sep=5em]
    \pazocal{A}(\theta) 
        \arrow[d,swap,"\mathrm{RH}"]
      \arrow[r,"s_i"]
    & \pazocal{A}\left(w_i(\theta)\right)
      \arrow[d,swap,"\mathrm{RH}"']
    \\
     \pazocal{S}(\omega) \arrow[r,"\mathrm{id}"] & \pazocal{S}(\omega)
    \end{tikzcd}
    \end{equation}
 commutes. Switch \eqref{omega} is necessary for this characterization in that local data are invariant for all but one generator of $W(\tilde{D}_4)$, visualized as the central node in the Dynkin diagram \ref{fig:aD4} for Okamoto's notation $w_2$, that instead preserves data $\omega$. In other words, the B\"acklund transformation $s_2$ alters local data $a$ but preserves the ``global'' ones $\{x,\omega\}$ coordinatizing the cubic.
 \begin{figure}[!t]
    \centering
    \includegraphics[height=2.65cm]{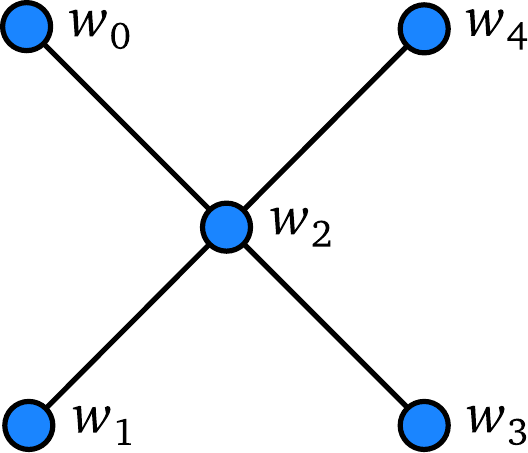}
    \caption{Okamoto's symmetries $w_i$ for $P_{V\!I}$ on the affine Dynkin diagram $\tilde{D}_4$.}\label{fig:aD4}
\end{figure}

This settles the understanding of $w_2$ in the moduli sense, but leaves the characterization problem open in terms of the unquotiented spaces:

\begin{problem}\label{pb1}
    Lift diagram \eqref{Masahiko} for $w_2$ to the connection and representation spaces, namely unravel the unknown arrows making the following diagram commute:
    \begin{equation}
        \begin{tikzcd}[row sep=large, column sep=5em]
    \pazocal{F}(\theta) 
        \arrow[d,swap,"\mathrm{RH}"]
      \arrow[r,"?"]
    & \pazocal{F}\!\left(w_2(\theta)\right)
      \arrow[d,swap,"\mathrm{RH}"']
    \\
     \pazocal{R}(\iota) \arrow[r,"?"] & \pazocal{R}\left(w_2(\iota)\right)
    \end{tikzcd}
    \end{equation}
\end{problem}

The main result of this paper is the explicit construction of the lower arrow. In light of the above considerations on group invariants, this is expected to be far from the trivial moduli counterpart.

A necessary working tool over the representation space \eqref{Repspace} is a coordinatization of the monodromy group: this rather nontrivial ingredient is provided, endowed with a crucial $\pazocal{X}$-cluster Poisson structure, by the (higher) Teichm\"uller machinery developed in \cite{DalMartello2024}.
Explicitly, the triple of independent generators corresponding to the finite punctures reads over the cluster coordinates $Z_{O2},Z_{B2},Z_{G2}$ as
\begin{equation}\label{Ms}
\begin{aligned}
    &M_1=\begin{pmatrix}
        0 & \iota_1^{-1}Z_{O2}^{-1}\\[.5em]
        -\iota_1 Z_{O2} & \iota_1+\iota_1^{-1}
    \end{pmatrix},\\
    &M_2=\begin{pmatrix}
        \iota_2+\iota_2^{-1}+\iota_2^{-1}Z_{B2}^{-1} & \iota_2+\iota_2^{-1}+\iota_2^{-1}Z_{B2}^{-1}+\iota_2Z_{B2}\\[.5em]
        -\iota_2^{-1}Z_{B2}^{-1} & -\iota_2^{-1}Z_{B2}^{-1}
    \end{pmatrix},\\
    &M_3=\begin{pmatrix}
        \iota_3+\iota_3^{-1}+\iota_3Z_{G2} & \iota_3Z_{G2}\\[.5em]
        -\iota_3-\iota_3^{-1}-\iota_3^{-1}Z_{G2}^{-1}-\iota_3Z_{G2} & -\iota_3Z_{G2}
    \end{pmatrix},
\end{aligned}
\end{equation}
and the cluster structure allows to distill Okamoto's transformation on $\pazocal{R}(\iota)$ into a single rational map of cluster charts, whose shape
\begin{equation}\label{tastemut}
    Z_\alpha \longmapsto \frac{1+Z_\alpha+Z_\alpha Z_\beta}{Z_\beta(1+Z_\gamma+Z_\alpha Z_\gamma)}
\end{equation}
can be purely encoded in the language of $\pazocal{X}$-mutations.

Before detailing its cluster combinatorics, the above formula must be stressed to be far from Laurent, thus inducing ``unorthodox'' transformations on monodromy matrices \eqref{Ms}: entries turn truly rational, eluding the customary restriction to universally Laurent elements---which define the regular space of functions for the so-called cluster $\pazocal{X}$-variety \cite{Fock2009}. Nevertheless, being encapsulated as a map of $\pazocal{X}$-coordinates by just a sequence of mutations, $w_2$ induces on the geometry of the cluster variety an isomorphism of cluster structures. 

\begin{remark}
    Overstepping Laurentness is generally thought of as a shortcoming, and yet it is precisely the leap needed to capture $w_2$ on monodromy \emph{matrices}. Indeed, unlike coordinates \eqref{mon_mfd}, entries of the $\pazocal{X}$-coordinatization overstep the regular $\pazocal{X}$-variety in the first place by failing to be Laurent under any mutation, and this fact stems from the need to capture the entire matricial information of the monodromy group, i.e., handle the representation variety $\pazocal{R}(\iota)$ in place of the usual monodromy manifold $\pazocal{S}(\omega)$.
\end{remark}

For $\mu_*$ the mutation at coordinate $Z_*$, the rational map of $\pazocal{X}$-coordinates is indeed understood in cluster terms via the recognizable (cf. \Cref{rmk:reflection}) mutation formula
\begin{equation}\label{mutprel}
\mu_{w_2}:=\mu_\beta\mu_\gamma\mu_\beta\mu_\gamma\mu_\beta\mu_\alpha\mu_\gamma\mu_\beta\mu_\alpha,
\end{equation}
bringing a wealth of combinatorics into the picture.

To start with, in the quiver-theoretic terms of classical cluster algebra, the structure ruling mutations of our $\pazocal{X}$-coordinatization is of type $A_3$: essentially, the quiver resulting from the Teichm\"uller machinery is a $3$-cycle (\Cref{fig:2amalquiver}). In this language, the sequence of mutations \eqref{mutprel} is singled out by leaving such quiver invariant.
In Painlevé-theoretic terms, mutations manifest in the language of fat graph flips. For $P_{V\!I}$ and its punctured domain $\Sigma_{0,4}$, the fat graph structure is $3$-star-shaped (\Cref{fig:operations}) and its geometry encodes the cluster dynamics of $\pazocal{X}$-coordinates, which are in one-to-one correspondence with the fat graph's edges, by giving a dictionary between flip $\alpha$ and mutation $\mu_\alpha$.
Thus, with unorthodox transformations come unorthodox flips: in order to capture $w_2$'s mutation formula, we allow fat graph flips even on edges incident to a loop---providing a companion recipe for the transformation of the loop's corresponding coordinate. This oversteps the restriction of the current framework, flips of which only preserve loops (cf. Figure 3 in \cite{Chekhov2020}) and lead to generalized $\pazocal{A}$-mutations.
We name the geometric operation resulting from sequence \eqref{mutprel} under this new set of rules as \textbf{inside-out}, after its reversal effect on the star-shape visualized in \Cref{fig:operations}.

\begin{figure}[!t]
    \centering
    \includegraphics[width=.6\textwidth]{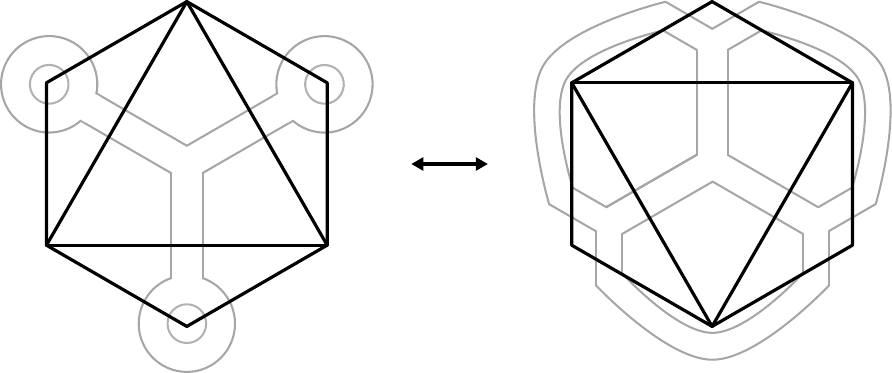}
    \caption{Dual characterizations for the monodromic $w_2$: the $\pi$-rotation on equilateral triangulations of the hexagon $\{6\}$ and the inside-out on star-shaped fat graphs of the $4$-punctured Riemann sphere $\Sigma_{0,4}$, namely the domain of the Fuchsian representation \eqref{Sysspace}.}\label{fig:operations}
\end{figure}


Dualizing fat graphs, the inside-out operation translates to the language of triangulations. In fact, the resulting dual characterization of $w_2$'s mutation formula is more essential and allows to be faithfully codified by the polygonal shape of the colored associahedron.
Indeed, the set of star-shaped fat graphs of the $4$-punctured sphere $\Sigma_{0,4}$ dualizes to that of triangulations of the hexagon $\{6\}$, and fat graph flips translate to flips of triangulations. In particular, chart $\{Z_{O2},Z_{B2},Z_{G2}\}$ is attached to the triangulation's defining triple of chords.
The geometry underlying the combinatorics of flips of triangulations for the hexagon is well-known to be given by the so-called $3$-dimensional associahedron $\mathbb{A}_3$.  
However, for flips to succeed in capturing the exact rational formula for $w_2$ through the corresponding mutations, we must be able to tell apart the dynamics of \emph{individual} cluster coordinates.
We achieve this refinement by passing to the \textbf{colorful associahedron} $\mathbb{A}_3^c$, whose chords are painted with mutually distinct colors from a selected palette of three. 
In particular, the dynamics of each cluster variable, now attached to a colored chord, is fully captured by the combinatorics of flips on $\mathbb{A}^c_3$ and $w_2$ is singled out as the $\pi$-rotation on equilateral triangulations visualized in \Cref{fig:operations}, which showcases the compatibility with the inside-out.

\begin{remark}
    In cluster algebra lingo, painting the triangulation corresponds to consider labeled seeds instead of usual up-to-permutation classes. The exchange graph for type $A_n$, whose vertices correspond to equivalence classes of seeds, is indeed given by the $1$-skeleton of Stasheff's $n$-associahedron \cite{Fomin2021}.
    The colored $3$-associahedron $\mathbb{A}_3^c$ is precisely a covering of the standard one by the symmetric group, i.e., reduces to the latter ``by going color blind''. For finite type cluster algebras, this is an invitation to consider colorful generalized associahedra as the natural geometric locus for labeled seeds. In the same spirit, Ishibashi and Kano \cite{Ishibashi2024} advocate for a labeled approach with the notion of ``labeled exchange graph''.
\end{remark}

All combined, our nonstandard approach culminates in the following main result:
\begin{theorem}
    Okamoto's transformation $w_2$ lifts to the representation space as the entry-wise action of mutation formula
    \begin{equation*}
        \begin{matrix} \mu_{w_2} & : & \pazocal{R}(\iota) & \longrightarrow & \pazocal{R}\left(w_2(\iota)\right)
        \end{matrix}
    \end{equation*}
    which admits dual geometric characterizations as the $\pi$-rotation on colored triangulations of the hexagon and the inside-out on star-shaped fat graphs of the four-punctured Riemann sphere.
\end{theorem}

This solves the monodromic facet of \Cref{pb1}. As it happens, the remaining of the diagram naturally unfolds when switching to the language of middle convolutions.

In terms of Fuchsian isomonodromy representations \eqref{Sysspace}, all generators but $w_2$ are understood as elementary gauge transformations in $\Gamma$. The analogue of $s_2$ on space $\pazocal{F}(\theta)$ has been originally unraveled by Filipuk \cite{Filipuk2006} through the additive \textbf{middle convolution} $mc_\xi$.
This functorial operation is tailored so to map $\mathbf{A}:=(A_1,A_2,A_3)\in \mathfrak{sl}_2(\mathbb{C})^{\times3}$ to a new triple of same dimension but with \emph{shifted} defining parameters $\theta$ exactly à la $w_2$.
In turn, formula \eqref{tastemut} itself finds theoretical genesis in convolutional terms.
The middle convolution toolkit has a multiplicative analogue $\MC_{\nu}$, which allows to extend Filipuk's result from the framework of Fuchsian systems to that of monodromy groups: the functor maps between triples $\mathbf{M}:=(M_1,M_2,M_3)\in SL_2(\mathbb{C})^{\times3}$ by \emph{scaling} the spectral parameters $\iota$ as commanded by $w_2$. Performing such operation on the right basis indeed recovers the whole entry-wise action of \eqref{tastemut}.

With both realizations of $w_2$ speaking the same convolutional language, correspondence \eqref{MCRH} at the core of the theory ensures they precisely commute with the Riemann-Hilbert map: for the compatible specializations of the convolutional parameters
\begin{equation*}
    \nu=e^{2\pi i \vartheta},\quad \vartheta:=-\frac{\theta_1+\theta_2+\theta_3+\theta_\infty}{2},
\end{equation*}
we finally obtain the commutative square
    \begin{equation}\label{pb1sol}
        \begin{tikzcd}[row sep=large, column sep=5em]
    \pazocal{F}(\theta) 
        \arrow[d,swap,"\mathrm{RH}"]
      \arrow[r,"mc_\vartheta"]
    & \pazocal{F}\left(w_2(\theta)\right)
      \arrow[d,swap,"\mathrm{RH}"']
    \\
     \pazocal{R}(\iota) \arrow[r,"\MC_{e^{2\pi i \vartheta}}"] & \pazocal{R}\left(w_2(\iota)\right)
    \end{tikzcd}
    \end{equation}
solving \Cref{pb1} in its entirety.

This does not mark the end of the story though, since a realization of $w_2$ parallel to Filipuk's is available for the alternative Birkhoff framework of $P_{V\!I}$.

Via Harnad's \textbf{duality} $\mathfrak{H}^{\!\vee}$ \cite{Harnad1994}, the sixth Painlevé equation indeed admits an equivalent Birkhoff isomonodromy representation
\begin{equation}\label{irr}
\begin{aligned}
    \pazocal{B}(\theta)=\left\{ {\mathrm d}-\left(U+\frac{V-\One}{z}\right){{\mathrm d}z} \,\Big|\,U=\mathrm{diag}(0,1,t);\, V_{kk}=-\theta_k,\,\mathrm{eig}(V)=\Big\{0,\textstyle\frac{\pm\theta_\infty-\theta_1-\theta_2-\theta_3}{2}\Big\}\right\},
\end{aligned}
\end{equation}
whose generalized monodromy data, singled out by a genuine monodromy matrix at $0$ and a pair of Stokes ones at $\infty$ \cite{Mazzocco2002}, are characterized by triples
\begin{equation}\label{eq:irrdata}\begin{aligned}
\Big\{(M_0,S_1,S_2) \in GL_3(\mathbb C)\times B_+^{(1)}\times  B_-\,\big|\,& M_0S_1S_2=\One,\\&\mathrm{eig}(M_0)=\big\{1,e^{\pi i(\theta_\infty-\theta_1-\theta_2-\theta_3)},e^{-\pi i(\theta_\infty+\theta_1+\theta_2+\theta_3)}\big\},\\
&\mathrm{eig}(S_2)=\big\{e^{2\pi i\theta_1},e^{2\pi i\theta_2},e^{2\pi i\theta_3}\big\}\Big\},
\end{aligned}
\end{equation}
for the Borel subgroups $B_+^{(1)}$ of upper unitriangular matrices and $B_-$ of lower triangular ones.
Mazzocco \cite{Mazzocco2004} showed that, in this representation, the \emph{whole} group $W(\tilde{D}_4)$ of symmetries sits inside the gauge one $\Gamma$.

\begin{problem}\label{pb2}
    Develop a unified description for $w_2$ on connection and representation spaces that encompasses both the Birkhoff and Fuchsian frameworks in a Riemann-Hilbert compatible way.
\end{problem}

Again, the crucial ingredient to attack the problem is provided by the theory developed in \cite{DalMartello2024}, and the solution manifests naturally in the middle convolution language. 

Besides the Teichmüller machinery delivering the crucial $\pazocal{X}$-coordination \eqref{Ms}, \cite{DalMartello2024} introduces the \textbf{GDAHA functor} $\mathscr{F}_q$ as a map between representation categories of generalized double affine Hecke algebras of uniform length (cf. \Cref{rmk:MCq}).
Taking the $q\rightarrow1$ classical limit, the functor simplifies to a two-step operation between matrix tuples: a multiplicative middle convolution $\MC$ followed by a classical result we refer to as the \textbf{Killing factorization}.
Writing $\mathfrak{H}^{\!\vee}$ in convolutional language, Fuchsian and Birkhoff formulations turn equivalent also in terms of (generalized) monodromy, with the identification provided precisely by $\mathscr{F}_1$: for the middle convolutions acting on preconditioned triples $\widehat{\mathbf{A}}$ and $\widehat{\mathbf{M}}$ as is customary in the theory,
\begin{equation}\label{square}
        \begin{tikzcd}[row sep=large, column sep=5em]
    \frac{\mathrm{d}}{\mathrm{d}\lambda}\Phi=\big(\sum^3_{k=1}\frac{{\widehat{A}}_k}{\lambda-u_k}\big)\Phi 
        \arrow[d,swap,"\mathrm{RH}"]
      \arrow[r,"\mathfrak{H}^{\!\vee}"]
    & \frac{\mathrm{d}}{\mathrm{d}z}Y=\left(U+\frac{V-\One}{z}\right)Y
      \arrow[d,swap,"\mathrm{RH}"']
    \\
     (\widehat{M}_1,\widehat{M}_2,\widehat{M}_3) \arrow[r,"\mathscr{F}_1"] & (S_1,S_2)
    \end{tikzcd}
    \end{equation}
commutes and reads as an holistic definition of $P_{V\!I}(\theta)$ through its pair of convolutional-compatible isomonodromy representations.

Thanks to this identification, we solve \Cref{pb2} by framing $w_2$ as a $4$-tuple of maps between two such \textbf{Painlevé squares} \eqref{square} at different values of the parameters, encompassing the differential results of Filipuk and Mazzocco in a unified cube-shaped diagram (\Cref{fig:cube}). Among the four, it is indeed our cluster monodromic map that showcases the deepest and most ramified connections.   

The present paper is organized as follows.

\Cref{sec:middleconv} recaps basic definitions and properties of both versions of the middle convolution and their corresponding preconditioners, specifying the adopted notation.

The main \Cref{sec:Okamoto} is dedicated to the cluster monodromic realization of Okamoto's birational transformation and its combinatorial features. We introduce such $w_2$ via the multiplicative middle convolution, and detail first its simpler characterization on triangulations.

\Cref{sec:cube} frames this new realization in the context of the Painlevé VI duality, relying on the diagrammatic language. By providing the computational details, it serves as a technical companion to Appendix B of \cite{DalMartello2024} as was there anticipated.

Finally, \Cref{app:X&A} gives a minimal primer on cluster ensembles.

\paragraph{Acknowledgments}
\!\!\!The author is grateful to Marta Mazzocco for suggesting this line of research and participating in many insightful discussions. This project was funded by the Engineering and Physical Sciences Research Council [2438494]; the Japan Society for the Promotion of Science [PE24720]; and Fondazione Cariparo [C93C22008360007] through the University of Padua.

\section{Middle convolution and generalizations}\label{sec:middleconv}

This preliminary section defines the middle convolution in both its additive and multiplicative versions, together with its preconditioning addition functor. En passant, relevant generalizations of these operations are briefly discussed in the form of remarks.

\subsection{Multiplicative version}\label{sec:multconv}

Katz \cite{Katz1996} introduced the middle convolution functor to prove an existence theorem for irreducible rigid local systems.
Any such system was shown to originate from the trivial one $\mathrm{d}\phi=0$ by applying invertible sequences of preconditioned middle convolutions, leading simultaneously to a classification and an existence algorithm.
The functor preserves important properties like the index of
rigidity and irreducibility, but in general changes the rank and the monodromy group.

Following \cite{Dettweiler2007}, we give a purely algebraic analogue, that reproduces the functor's core properties, as the endofunctor
\begin{equation}
    \MC_{\nu} : \mathsf{Mod}(\mathbb{C}[F_p]) \longrightarrow \mathsf{Mod}(\mathbb{C}[F_p]),
\end{equation}
where $\nu\in\mathbb{C}^*$ and $\mathsf{Mod}(\mathbb{C}[F_p])$ is the category of finite-dimensional (left) $\mathbb{C}[F_p]$-modules, $F_p$ denoting the free group on $p$ generators. More transparently, objects in $\mathsf{Mod}(\mathbb{C}[F_p])$ can be viewed as couples $(\mathbf{M},V)$,
\begin{equation*}
    \mathbf{M}=(M_1,M_2,\ldots,M_p)\in GL(V)^p,
\end{equation*}
where each matrix represents the action of the respective generator on the vector space $V$.
We can thus detail the functor as a map $(\mathbf{M},V)\longmapsto(\widetilde{\mathbf{N}},W)$, $\widetilde{\mathbf{N}}\in GL(W)^p$, between $p$-tuples of matrices---in particular, monodromy ones.
  
The intermediate object $(C_{\nu}(\mathbf{M}),V^p)\in  \mathsf{Mod}(\mathbb{C}[F_p])$,
\begin{equation*}
    C_{\nu}(\mathbf{M})=(N_1,\ldots,N_p)\in GL(V^p)^p,
\end{equation*}
is first defined by formulae 
\begin{equation}
N_i = \begin{pmatrix}
                  1 & 0 &  & \ldots& & \\
                  0 & \ddots &  & & &\\
                   \vdots & & 1 &&&\\
                \nu(M_1-\One) & \ldots&  \nu(M_{i-1}-\One)  & \nu M_{i} & M_{i+1}-\One & \ldots 
&   M_p-\One \\
     &  && & 1 & & \vdots \\
              &   &  & &  & \ddots  &  0 \\             
                    &  &  & \ldots& & 0 & 1
          \end{pmatrix}.
\end{equation}
In order to preserve rigidity and irreducibility, the middle convolution is then obtained as the restriction of this enlarged tuple on the quotient $V^p/(\pazocal{K}+\pazocal{L})$, where
\begin{equation*}
\pazocal{K}:=\bigoplus_{i=1}^p\pazocal{K}_i, \qquad
    \pazocal{K}_i = \left( \begin{array}{c}
          0 \\
          \vdots \\
          0 \\
          \ker(M_i-\One) \\
          0 \\
          \vdots\\
          0
        \end{array} \right)  \quad \mbox{($i$-th entry)},
\end{equation*}
and 
\[
\pazocal{L}=\bigcap_{i=1}^p \ker (N_i-\One)=\mathrm{ker}(N_1\cdots N_p - \One)
\]
are $\langle N_1,\ldots,N_p \rangle$-invariant 
subspaces of $V^p$.
\begin{definition}
The object
$(C_{\nu}(\mathbf{M}),V^p)$ is the \textbf{convolution} of $\mathbf{M}$.
The object $(\MC_{\nu}(\mathbf{M})$,$V^p/ (\pazocal{K}+\pazocal{L}))$ is the \textbf{middle convolution} of $\mathbf{M}$, where the matrix tuple $$\MC_{\nu}(\mathbf{M}):=\left(\widetilde{N}_1,\dots,\widetilde{N}_p\right)\in GL(V^p/(\pazocal{K}+\pazocal{L}))^p$$ has each $\widetilde{N}_k$ induced by the action of the corresponding element of $C_{\nu}(\mathbf{M})$ on the quotient.
\end{definition}
Among its many properties, the functor is multiplicative \cite[Theorem 2.4]{Dettweiler2007} and, for $\nu\neq1$, $\pazocal{K}+\pazocal{L}=\pazocal{K}\oplus\pazocal{L}$ via the explicit characterization
\begin{equation*}
\pazocal{L}=\Bigg\langle \ \left( \begin{array}{c}
          M_{2} \cdots M_p\mathbf{v} \\
          M_{3} \cdots M_p\mathbf{v} \\
          \vdots \\
          \mathbf{v}
        \end{array} \right)
        \ \Bigg| \ \mathbf{v} \in \mathrm{ker}(\nu\,M_1\cdots M_p-\One) \ \Bigg\rangle.
\end{equation*}
\begin{remark}\label{rmk:MCq}
    The quantum analogue of the multiplicative middle convolution was introduced in \cite{DalMartello2024}, for the \textbf{parameter-free} formulation $\MC$ whose definition excludes the subspace $\pazocal{L}$, and denoted by $\mathscr{M}_q$. 
    Postcomposed with (quantum) Killing factorization, $\mathscr{M}_q$ proves that quantized Stokes data form a representation for the generalized double affine Hecke algebra (GDAHA) of type $\tilde{E}_6$.
    This gives the irregular setting a quantum mirror of the Fuchsian one, since the monodromy group was known to quantize as a representation of the GDAHA of type $\tilde{D}_4$ \cite{Mazzocco2018}.
    Denoting by $H_{\mathscr{D}}$ the GDAHA of type $\mathscr{D}$, we can visualize the whole quantization as follows:
    \begin{equation}\label{finaldiagram}
    \begin{tikzcd}[row sep=large, column sep=7em]
     \arrow[rr,bend right=-15,dashed,"\mathscr{F}_1"']
     (\widehat{M}_1,\widehat{M}_2,\widehat{M}_3) \arrow[r,"\MC"'] & (R_1,R_2,R_3) \arrow[r,"R_1R_2R_3=S_1S_2"'] & (S_1,S_2)\\
  \arrow[rr,bend right=15,2blue,dashed,"\mathscr{F}_q"]
 \text{\color{2blue}$(\widehat{M}^q_1,\widehat{M}^q_2,\widehat{M}^q_3)$}  \arrow[r,2blue,"\mathscr{M}_q"] \arrow[u,swap,"q\rightarrow1"] & \text{\color{2blue}$(R^q_1,R^q_2,R^q_3)$}\arrow[r,2blue,"R^q_1R^q_2R^q_3=S^q_1S^q_2"] \arrow[u,swap,"q\rightarrow1"] & \text{\color{2blue}$(S_1^q,S_2^q)$} \arrow[u,swap,"q\rightarrow1"]\\[-3.2em]
 \color{mygray}\rotatebox[origin=c]{-90}{$\in$} & & \color{mygray}\rotatebox[origin=c]{-90}{$\in$}\\[-3.25em]
\color{mygray}\scalebox{0.75}{$\mathrm{Rep}\big(H_{\tilde{D}_4}\big)$} & & \color{mygray}\scalebox{0.75}{$\mathrm{Rep}\big(H_{\tilde{E}_6}\big)$}
\end{tikzcd}
\end{equation}
We refer to \cite[Appendix B]{DalMartello2024} for further insight.
\end{remark}

\subsection{Additive version}

The functor admits a parallel version for Fuchsian systems \cite{Dettweiler2007}, the two related by a Riemann-Hilbert correspondence: the latter becomes the former when passing to monodromy.
Essentially, the operation is the additive analogue: for $\xi\in\mathbb{C}$, $i=1,\ldots,p$ and $\mathbf{A}=(A_1,\ldots,A_p)$ the tuple of $n \times n$ finite matrix residua of a Fuchsian system \eqref{Fuchn},  construct first the block matrices  
\[
B_i:= \left( \begin{array}{ccccccc}
                   0 & \ldots  &  & & & \\
                   \vdots & \ddots &  & & &\\
               A_1 & \ldots & A_{i-1}&  A_i+ \xi\One  &  A_{i+1} & \ldots 
&  A_p \\
               &   &  &   & \ddots & \vdots  \\             
                   &    & & & \ldots & 0
          \end{array} \right) \in \mathbb{C}^{np \times np},
\]
each one vanishing outside the corresponding $i$-th block row.
Then, take the quotient over the two $\langle B_1,\ldots,B_p \rangle$-invariant 
subspaces of $\mathbb{C}^{np}$:
\begin{align*}
    &{\mathscr{K}}=\bigoplus_{i=1}^p {\mathscr{K}}_i, \qquad
    {\mathscr{K}}_i = \left( \begin{array}{c}
          0 \\
          \vdots \\
            0\\
          \ker(A_i) \\
           0 \\
           \vdots \\
           0 
        \end{array} \right) \quad \mbox{($i$-th entry)},\\[1em]
        &{\mathscr{L}}=\bigcap_{i=1}^p \mathrm{ker}(B_i)=\mathrm{ker}(B_1+\ldots+B_p).
\end{align*}
\begin{definition}
The matrix tuple $c_{\xi}(\mathbf{A}):=(B_1,\dots,B_p)$ is the \textbf{(additive)
 convolution}
of $\mathbf{A}$. The matrix tuple 
$mc_{\xi}:=(\widetilde{B}_1,\dots,\widetilde{B}_p)\in \mathbb{C}^{m\times m}$ is the
\textbf{(additive) middle convolution} of $\mathbf{A}$,
where each $\widetilde{B}_i$ is induced by the action of the corresponding element of $c_{\xi}(\mathbf{A})$ on 
$\mathbb{C}^m \simeq \mathbb{C}^{np}/({\mathscr{K}}+{\mathscr{L}})$.
\end{definition}
Mirroring the multiplicative case, $\xi\neq0$ yields $\mathscr{K}+\mathscr{L}=\mathscr{K}\oplus\mathscr{L}$ via the explicit characterization
\begin{equation*}
\mathscr{L}=\Bigg\langle \ \left( \begin{array}{c}
        \mathbf{v} \\
        \vdots \\
        \mathbf{v}
        \end{array} \right)
        \ \Bigg| \ \mathbf{v} \in \ker(A_1 + \cdots + A_p + \xi\One) \ \Bigg\rangle.
\end{equation*}

Denoting as $\mathrm{Fu}_\mathbf{A}$ the $n$-dimensional Fuchsian system \eqref{Fuchn} defined by $\mathbf{A}=(A_1,\ldots,A_p) \in (\mathbb{C}^{n \times n})^p$, the functors' Riemann-Hilbert correspondence is stated \cite[Theorem  4.7]{Dettweiler2007} in the following
\begin{theorem}\label{thm:MCcorr}
Let $\mathbf{M}:= \mathrm{Mon}(\mathrm{Fu}_\mathbf{A})=(M_1,\ldots,M_p)\in GL_n(\mathbb{C})^p$ be the tuple of monodromy generators for $\mathrm{Fu}_\mathbf{A}$, $\xi \in \mathbb{C}\backslash\mathbb{Z}$ and $\nu=e^{2\pi i\xi}$. If $\mathbf{M}$ generates an irreducible subgroup of $GL(V)$ for at least two $M_i$s different from the identity and
\begin{equation}
\begin{aligned}
    \mathrm{rk}(A_i)&=\mathrm{rk}(M_i-\One),\\
    \mathrm{rk}(A_1+\ldots+A_p+\xi\One)&=\mathrm{rk}(\nu M_1\cdots M_p-\One),
\end{aligned}
\end{equation}
then
\begin{equation}\label{MCRH}
    \mathrm{Mon}\big(\mathrm{Fu}_{mc_{\xi-1}(\mathbf{A})}\big)=\MC_{\nu}(\mathbf{M}).
\end{equation}
\end{theorem}
This correspondence, proving that $\MC_{e^{2\pi i \xi}}$ is \emph{the} map between the monodromy matrices of the respective Fuchsian systems mapped by $mc_\xi$, is a foundational ingredient of the theory. E.g., it drives a general scheme to produce constructive solutions to
the Riemann–Hilbert problem \cite{Biblio2017}.
For us, it ensures that diagram \eqref{pb1sol} commutes.

When it comes to applications, $mc_\xi$ is tailored by tweaking the invariant subspaces: $\mathscr{L}$ via the value of $\xi$, $\mathscr{K}$ via the preconditioning \textbf{addition} functor
\begin{equation}\label{addfunctor}
    ad_{\boldsymbol\upsilon}:(A_1,\ldots,A_p) \longmapsto (\widehat{A}_1,\ldots,\widehat{A}_p)=(A_1+\upsilon_1\One,\ldots,A_p+\upsilon_p\One)
\end{equation}
depending on a vector parameter $\boldsymbol\upsilon=(\upsilon_1,\ldots,\upsilon_p)$.
The very same tailoring can be done for the multiplicative case, where the preconditioner translates to monodromy via the Riemann-Hilbert correspondence
\begin{equation}\label{RHprecond}
    ad_{\boldsymbol\upsilon} \ \longleftrightarrow \ AD_{\boldsymbol\tau},
\end{equation}
for
\begin{equation}\label{AD}
    AD_{\boldsymbol\tau}: (M_1,\ldots,M_p) \longmapsto (\widehat{M}_1,\ldots,\widehat{M}_p)=(\tau_1M_1,\ldots,\tau_pM_p)
\end{equation}
whose vector parameter ${\boldsymbol\tau}=(\tau_1,\ldots,\tau_p)$ matches $e^{2\pi i \boldsymbol\upsilon}\!:=(e^{2\pi i\upsilon_1},\ldots,e^{2\pi i\upsilon_p})$.
Thus, the functor $mc_\xi\circ ad_{\boldsymbol\upsilon}$ maps between Fuchsian systems of tailored dimensions, and the Riemann-Hilbert correspondences ensure $\MC_{e^{2\pi i \xi}}\circ AD_{e^{2\pi i \boldsymbol\upsilon}}$ maps between the respective monodromy data. 

\begin{remark}\label{Takemura}
    The additive middle convolution has been extended \cite{Takemura2012} to encompass irregular systems in form
    \begin{equation*}
        \frac{\mathrm{d}}{\mathrm{d}z}Y=\bigg(-\sum_{j=1}^{m_0}A_j^{(0)}z^{j-1}+\sum_{i=1}^p\sum_{j=0}^{m_i}\frac{A_j^{(i)}}{(z-u_i)^{j+1}}\bigg)Y.
    \end{equation*}
    On Fuchsian systems \eqref{Fuchn}, namely the case $m_i=0$, $i=0,\ldots,p$, it reduces to the standard middle convolution. This prevents from phrasing diagram \ref{fig:square} in purely convolutional terms, but invites to look for an Euler-integral origin of the Laplace transform itself and develop a unified framework.
    On Birkhoff systems \eqref{irr}, this extension gives a true generalization when paired with the preconditioner
    \begin{equation*}
        ad_{(\upsilon_1,\upsilon_0)}\,:\,(U,V) \longmapsto (U+\upsilon_1\One,V+\upsilon_0\One),
    \end{equation*}
    whose differential counterpart is the gauge transformation $Y\mapsto e^{z\upsilon_1}z^{\upsilon_0}Y$. On Stokes data, the shift on $U$ is undetected and this generalized addition reads as
    \begin{equation}
        (M_0,S_1,S_2) \longmapsto (\tau_0M_0,S_1,\tau_0^{-1}S_2),
    \end{equation}
    where $\tau_0:=e^{2\pi i\upsilon_0}$.
\end{remark}

\section{Main result}\label{sec:Okamoto}

In \cite{Filipuk2007}, Filipuk and Haraoka computed an additive convolutional realization of Okamoto’s $w_2$ for $P_{V\!I}$. The recipe requires a parametrization of the Fuchsian connection \eqref{Sysspace} with rank $1$ residua, readily achieved via addition with vector parameter $\frac{\boldsymbol\theta}{2}:=\left(\frac{\theta_1}{2},\frac{\theta_2}{2},\frac{\theta_3}{2}\right)$: indeed,   
\begin{equation}\label{hatA}
    \widehat{\mathbf{A}}:=(\widehat{A}_1,\widehat{A}_2,\widehat{A}_3)=ad_{\frac{\boldsymbol{\theta}}{2}}(\mathbf{A})
\end{equation}
has spectra
\begin{equation}
    \mathrm{eig}(\widehat{A}_k)=\left\{0,\theta_k\right\}.
\end{equation}
Notice that addition induces the map
\begin{equation}
    A_\infty\longmapsto\widehat{A}_\infty:=A_\infty-\frac{\theta_1+\theta_2+\theta_3}{2}\One.
\end{equation}

 Specializing $\xi\in\mathrm{eig}(\widehat{A}_\infty)$, the computation's core step $mc_\xi(\widehat{\mathbf{A}})$ results in a triple of $2\times2$ matrices with spectra $\{0,\theta_k+\xi\}$ \emph{shifted} precisely à la $w_2$---and whose isomonodromic deformation gives a solution to $P_{V\!I}$ that matches B\"acklund transformation $s_2$.
 
 This section gives this realization a multiplicative version: for $\widehat{\mathbf{M}}=AD_{e^{\pi i \boldsymbol{\theta}}}(\mathbf{M})$ and $\nu\in\mathrm{eig}(\widehat{M}_\infty)$, $\MC_{\nu}(\widehat{\mathbf{M}})$ results in a triple of $2\times2$ matrices with corresponding \emph{rescaled} spectra $\{1,\nu\cdot e^{2\pi i \theta_k}\}$. This multiplicative approach, run explicitly over cluster coordinates, has a crucial advantage: choosing a unique basis for the convolutional machinery allows to interpret Okamoto's symmetry beyond a change of parameters via a special sequence of cluster mutations endowed with remarkable combinatorial features.

\subsection{$\pazocal{X}$-coordinatization of the monodromy group}\label{sec:coordinatization}

The starting point for the convolutional explicit computation is a coordinatization of the Fuchsian monodromy group \eqref{Fuchdata}. This is a rather nontrivial ingredient of our construction, that we borrow from the higher Teichm\"uller machinery developed in \cite{DalMartello2024}---which, for our $2$-dimensional setting, in fact produces coordinates for the classical Teichm\"uller space as the group datum is set to the standard $PSL_2$.

For a selected hyperbolic surface (triangulated by its fat graph) and rank $n$ of the (classical) theory, the recipe sends a basis of loops to a tuple of $SL_n(\pazocal{X}_{\pazocal{Q}})$-matrices over a so-called $\pazocal{X}$-torus, whose formal definition is given in \Cref{app:X&A}. It is a split algebraic torus $\pazocal{X}_{\pazocal{Q}}$ endowed with a Poisson structure on cluster $\pazocal{X}$-coordinates $X_i$ such that $\{X_i,X_j\}=\varepsilon_{ij}X_iX_j$, for $(\varepsilon_{ij})$ the exchange matrix of a quiver $\pazocal{Q}$.
By choosing $n=2$ and the four-punctured Riemann sphere $\Sigma_{0,4}$, namely the domain of the Fuchsian system \eqref{Sysspace}, we obtain the following coordinatization under the choice of loops in \Cref{fig:2dimloops}:

\begin{theorem}\label{thm:D4-emb}
Let $\triangledown$ denote the quiver in \Cref{fig:2amalquiver} and $\{Z_{O1},Z_{O2},Z_{B1},Z_{B2},Z_{G1},Z_{G2}\}$ be the set of cluster coordinates for $\pazocal{X}_\triangledown$. Then, within $\mathrm{Mat}_2(\mathbb{C}[\pazocal{X}_\triangledown])$ the matrices
\begin{equation}\label{OBGP}
\begin{aligned}
    &\bar{O}=\begin{pmatrix}
        0 & Z_{O1}^{\nicefrac{-1}{2}}Z_{O2}^{-1}\\[.5em]
        -Z_{O1}^{\nicefrac{1}{2}}Z_{O2} & Z_{O1}^{\nicefrac{1}{2}}+Z_{O1}^{\nicefrac{-1}{2}}
    \end{pmatrix},\\
    &\bar{B}=\begin{pmatrix}
        Z_{B1}^{\nicefrac{1}{2}}+Z_{B1}^{\nicefrac{-1}{2}}+Z_{B1}^{\nicefrac{-1}{2}}Z_{B2}^{-1} & Z_{B1}^{\nicefrac{1}{2}}+Z_{B1}^{\nicefrac{-1}{2}}+Z_{B1}^{\nicefrac{-1}{2}}Z_{B2}^{-1}+Z_{B1}^{\nicefrac{1}{2}}Z_{B2}\\[.5em]
        -Z_{B1}^{\nicefrac{-1}{2}}Z_{B2}^{-1} & -Z_{B1}^{\nicefrac{-1}{2}}Z_{B2}^{-1}
    \end{pmatrix},\\
    &\bar{G}=\begin{pmatrix}
        Z_{G1}^{\nicefrac{1}{2}}+Z_{G1}^{\nicefrac{-1}{2}}+Z_{
        G1}^{\nicefrac{1}{2}}Z_{G2} & Z_{
        G1}^{\nicefrac{1}{2}}Z_{G2}\\[.5em]
        -Z_{G1}^{\nicefrac{1}{2}}-Z_{G1}^{\nicefrac{-1}{2}}-Z_{
        G1}^{\nicefrac{-1}{2}}Z_{G2}^{-1}-Z_{
        G1}^{\nicefrac{1}{2}}Z_{G2} & -Z_{
        G1}^{\nicefrac{1}{2}}Z_{G2}
    \end{pmatrix},\\
    &\bar{P}=\begin{pmatrix}
        Z_{O1}^{\nicefrac{1}{2}}Z_{B1}^{\nicefrac{1}{2}}Z_{G1}^{\nicefrac{1}{2}}Z_{O2}Z_{B2}Z_{G2} & 0\\
        -z & Z_{O1}^{\nicefrac{-1}{2}}Z_{B1}^{\nicefrac{-1}{2}}Z_{G1}^{\nicefrac{-1}{2}}Z_{O2}^{-1}Z_{B2}^{-1}Z_{G2}^{-1}
    \end{pmatrix}=\left(\bar{O}\,\bar{B}\,\bar{G}\right)^{-1},
\end{aligned}
\end{equation}
with
\begin{equation}\begin{aligned}
    z=&\,(Z_{O1}^{\nicefrac{1}{2}}-Z_{O1}^{\nicefrac{-1}{2}})Z_{B1}^{\nicefrac{-1}{2}}Z_{G1}^{\nicefrac{-1}{2}}Z_{B2}^{-1}Z_{G2}^{-1}+(Z_{B1}^{\nicefrac{1}{2}}-Z_{B1}^{\nicefrac{-1}{2}})Z_{G1}^{\nicefrac{-1}{2}}Z_{O1}^{\nicefrac{-1}{2}}Z_{O2}Z_{G2}^{-1}\\ &+(Z_{G1}^{\nicefrac{1}{2}}-Z_{G1}^{\nicefrac{-1}{2}})Z_{O1}^{\nicefrac{-1}{2}}Z_{B1}^{\nicefrac{-1}{2}}Z_{O2}Z_{B2}+Z_{O1}^{\nicefrac{1}{2}}Z_{B1}^{\nicefrac{1}{2}}Z_{G1}^{\nicefrac{1}{2}}Z_{O2}Z_{B2}Z_{G2}\\&+Z_{O1}^{\nicefrac{1}{2}}Z_{B1}^{\nicefrac{1}{2}}Z_{G1}^{\nicefrac{-1}{2}}Z_{O2}Z_{B2}Z_{G2}^{-1}+Z_{O1}^{\nicefrac{1}{2}}Z_{B1}^{\nicefrac{-1}{2}}Z_{G1}^{\nicefrac{-1}{2}}Z_{O2}Z_{B2}^{-1}Z_{G2}^{-1},
\end{aligned}
\end{equation}
satisfy the relations
\begin{equation}
\begin{aligned}
    \left(\bar{O}-Z^{\nicefrac{1}{2}}_{O1}\One\right)\left(\bar{O}-Z^{\nicefrac{-1}{2}}_{O1}\One\right)&=\Zero,\\
    \left(\bar{B}-Z^{\nicefrac{1}{2}}_{B1}\One\right)\left(\bar{B}-Z^{\nicefrac{-1}{2}}_{B1}\One\right)&=\Zero,\\
    \left(\bar{G}-Z^{\nicefrac{1}{2}}_{G1}\One\right)\left(\bar{G}-Z^{\nicefrac{-1}{2}}_{G1}\One\right)&=\Zero,\\
    \left(\bar{P}-Z_{O1}^{\nicefrac{1}{2}}Z_{B1}^{\nicefrac{1}{2}}Z_{G1}^{\nicefrac{1}{2}}Z_{O2}Z_{B2}Z_{G2}\One\right)\left(\bar{P}-Z_{O1}^{\nicefrac{-1}{2}}Z_{B1}^{\nicefrac{-1}{2}}Z_{G1}^{\nicefrac{-1}{2}}Z_{O2}^{-1}Z_{B2}^{-1}Z_{G2}^{-1}\One\right)&=\Zero,\\
        \bar{O}\,\bar{B}\,\bar{G}\,\bar{P}&=\One.
\end{aligned}
\end{equation}
Moreover, up to global conjugation and cyclic permutation, complex values of the cluster coordinates exist for any irreducible element of \eqref{Fuchdata} to satisfy \eqref{OBGP} as
\begin{equation}\label{MiasOBG}
    M_1 = \bar{O}, \quad M_2 = \bar{B}, \quad M_3 = \bar{G}, \quad M_\infty = \bar{P}. 
\end{equation} 
\end{theorem}
\begin{proof}
    Take the classical limit $q\rightarrow1$ in \cite[Theorem 12]{DalMartello2024}. The parametrization property follows from a geodesic argument as in \cite[§2.3]{Mazzocco2018}. 
\end{proof}

\begin{figure}
    \centering
    \includegraphics[width=.45\textwidth]{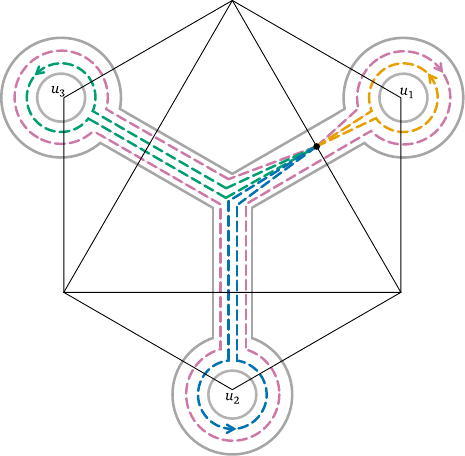}
    \caption{Basis of loops ({\color{3ochre}ochre}, {\color{3blue}blue} and {\color{3green}green}) on $\Sigma_{0,4}$ for the quadruple \eqref{OBGP}, whose matrices are named after the respective loop's color: e.g., $\bar{O}$ corresponds to the {\color{3ochre}ochre} loop encircling $u_1$, thus playing the role of $M_1$. The triangular shape mimics that of the quiver $\triangledown$, whose isolated vertices generate at punctures $u_i$ as eigenvalues of the monodromy matrices. The hexagonal triangulation results from dualizing the underlying {\color{mygray}gray} fat graph $\Gamma_{0,4}$, see \cite[§4.1]{DalMartello2024}. It is a crucial ingredient for the coordinatization, whose formulae are constructed from matricial building blocks associated to crossings of triangles as detailed in \cite[§3.2]{DalMartello2024}.}\label{fig:2dimloops}
\end{figure}

\begin{remark}
This coordinatization recovers the one constructed in \cite{Mazzocco2018} over shear coordinates, with the major advantage of having a natural $\pazocal{X}$-cluster structure. 
\end{remark}

The Poisson structure admits four evident Casimir elements, detectable from $\triangledown$'s three isolated vertices and isolated $3$-cycle (\Cref{fig:2amalquiver}).
\begin{figure}
    \centering
    \includegraphics[height=3.75cm]{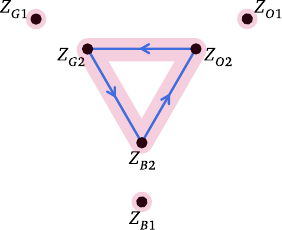}
    \caption{The triangular-shaped quiver $\triangledown$ encoding the quadratic Poisson structure, whose nonzero brackets are given by $\{Z_{B2},Z_{O2}\}=Z_{B2}Z_{O2}$ and its alphabetic cyclic permutations of subscripts. The highlighted $Z_{O1}$, $Z_{B1}$, $Z_{G1}$, and  $Z_{O2}Z_{B2}Z_{G2}$, generate the subalgebra of Casimir elements.
    }\label{fig:2amalquiver}
\end{figure}
It follows from the spectral specialization in \eqref{Repspace} that they match the $\theta$s as follows:
\begin{equation}\label{Zthetacorr}
    Z_{O1} =\iota_1^{-2}=e^{-2\pi i \theta_1},\quad
    Z_{B1} =\iota_2^{-2}=e^{-2\pi i \theta_2},\quad
    Z_{G1} =\iota_3^{-2}=e^{-2\pi i \theta_3};
\end{equation}
\begin{equation}\label{cyclecasimir}
    Z_{O2}Z_{B2}Z_{G2}=\iota_1\iota_2\iota_3\iota_\infty=e^{\pi i(\theta_\infty+\theta_1+\theta_2+\theta_3)}.
\end{equation}
In the following, we enforce all four equalities by taking the quotient of the Poisson algebra $\mathbb{C}[\pazocal{X}_\triangledown]$ with respect to the ideal generated by the corresponding Casimir evaluations. Thus, the resulting matrix quadruple simplifies over
\begin{equation}\label{simpalgebra}
    \mathbb{C}_{\iota}[\pazocal{X}_\triangledown]:=\mathbb{C}\big[ Z_{O2}^{\pm1},Z_{B2}^{\pm1},Z_{G2}^{\pm1}\big]\big/(Z_{O2}Z_{B2}Z_{G2}-\iota_1\iota_2\iota_3\iota_\infty).
\end{equation}

\subsection{Monodromic realization of $w_2$}

Having constructed a coordinatization for the monodromy matrices, we proceed with the explicit convolutional computations.

In order to mimic a symmetry, the multiplicative middle convolution must send the input triple $\mathbf{M}=(M_1,M_2,M_3)=(\bar{O},\bar{B},\bar{G})$ back to $\mathrm{Mat}_2(\mathbb{C})^{\times3}$.
This requires to maximize both invariant subspaces, achieving a $4$-dimensional sum $\pazocal{K}\oplus\pazocal{L}$. 
On the one hand, we tailor dim$(\pazocal{K})=3$ via the multiplicative preconditioner $AD_{e^{\pi i \boldsymbol\theta}}$ \eqref{AD} with vector parameter $e^{\pi i \boldsymbol\theta}=(e^{\pi i \theta_1},e^{\pi i \theta_2},e^{\pi i \theta_3})$, obtaining the following rescaled monodromy matrices over the evaluated cluster Poisson algebra $\mathbb{C}_{\iota}[\pazocal{X}_\triangledown]$:
\begin{equation}\label{rescOBG}
\begin{aligned}
    \widehat{M}_1=e^{\pi i \theta_1}\bar{O}&=\begin{pmatrix}
        0 & e^{2\pi i \theta_1}Z_{O2}^{-1}\\[.5em]
       -Z_{O2} & 1+e^{2\pi i \theta_1}
    \end{pmatrix},\\
    \widehat{M}_2=e^{\pi i \theta_2}\bar{B}&=\begin{pmatrix}
        1+e^{2\pi i \theta_2}+e^{2\pi i \theta_2}Z_{B2}^{-1} & 1+e^{2\pi i \theta_2}+e^{2\pi i \theta_2}Z_{B2}^{-1}+Z_{B2}\\[.5em]
        -e^{2\pi i \theta_2}Z_{B2}^{-1} & -e^{2\pi i \theta_2}Z_{B2}^{-1}
    \end{pmatrix},\\
    \widehat{M}_3=e^{\pi i \theta_3}\bar{G}&=\begin{pmatrix}
        1+e^{2\pi i \theta_3}+Z_{G2} & Z_{G2}\\[.5em]
        -1-e^{2\pi i \theta_3}-e^{2\pi i \theta_3}Z_{G2}^{-1}-Z_{G2} & -Z_{G2}
    \end{pmatrix},
\end{aligned}
\end{equation}
with spectra 
\[
\mathrm{eig}(\widehat{M}_k)=\{1,e^{2\pi i \theta_k}\}.
\]
On the other hand, we induce a nontrivial subspace $\pazocal{L}$ by setting $\nu=(Z_{O2}Z_{B2}Z_{G2})^{-1}$.
This is the eigenvalue $e^{-\pi i(\theta_\infty+\theta_1+\theta_2+\theta_3)}$ of $\widehat{M}_\infty=e^{-\pi i(\theta_1+\theta_2+\theta_3)}M_\infty$ for the simplest eigenvector $\mathbf{v}=(0,1)^T$ and, in Poisson terms, the cubic generator of the subalgebra of Casimir elements. In particular,
\begin{equation}    \pazocal{L}=\left\langle\begin{pmatrix}\widehat{M}_2\widehat{M}_3\mathbf{v}\\\widehat{M}_3\mathbf{v}\\\mathbf{v}\end{pmatrix}\right\rangle=\left\langle\begin{pmatrix}-Z_{B2}Z_{G2}\\0\\Z_{G2}\\-Z_{G2}\\0\\1\end{pmatrix}\right\rangle.
\end{equation} 

We then construct a change of variables by completing a basis of $\pazocal{K}\oplus\pazocal{L}$:
\begin{equation}\label{changeofvarOka}
    C_{Z}:=\begin{pmatrix}
                e^{2\pi i \theta_1}Z_{O2}^{-1} & 0 & 0 & -Z_{B2}Z_{G2} & a & b\\
                1 & 0 & 0 & 0 & 0 & 0\\
                0 & -1-e^{-2\pi i \theta_2}Z_{B2}^{-1} & 0 & Z_{G2} & 0 & 0\\
                0 & 1 & 0 & -Z_{G2} & 0 & 0\\             
                0 & 0 & -1 & 0 & 0 & 0\\
                0 & 0 & 1+e^{2\pi i \theta_3}Z_{G2}^{-1} & 1 & c & d
    \end{pmatrix},
\end{equation}
whose first three columns give a basis of $\pazocal{K}$, the fourth one generates the $1$-dimensional $\pazocal{L}$, and
\begin{equation}\label{abcd}\resizebox{\textwidth}{!}{$
    \begin{pmatrix}
        a & b \\ c & d
    \end{pmatrix}=\begin{pmatrix}
        -e^{\pi i (\theta_1+\theta_2)}Z_{B2}Z_{G2}^2(1+Z_{G2}^{-1})(1+Z_{B2}^{-1}+Z_{O2}^{-1}Z_{B2}^{-1}) & -e^{\pi i (\theta_1+\theta_2)}Z_{G2}^2(1+Z_{G2}^{-1}+Z_{O2}^{-1}Z_{G2}^{-1}) \\[.5em] e^{\pi i (\theta_1+\theta_2)}(1+Z_{B2}^{-1}+Z_{O2}^{-1}Z_{B2}^{-1}) & -e^{\pi i (\theta_1+\theta_2)}Z_{G2}(1+Z_{G2}^{-1}+Z_{O2}^{-1}Z_{G2}^{-1})
    \end{pmatrix}.$}
\end{equation}

Extracting the $2\times2$ lowest diagonal blocks in the $3$-tuple $(C_Z^{-1}N_1C_Z,C_Z^{-1}N_2C_Z,C_Z^{-1}N_3C_Z)$, for $(N_1,N_2,N_3)$ the convolution of $\widehat{\mathbf{M}}$ with parameter $Z_{O2}^{-1}Z_{B2}^{-1}Z_{G2}^{-1}$, we finally obtain
\begin{equation}\label{tildeN}\resizebox{\textwidth}{!}{$
\begin{aligned}
    \widetilde{M}_1=&\begin{pmatrix}
        0 & \frac{e^{2\pi i \theta_1}(1+Z_{G2}+Z_{O2}Z_{G2})}{Z_{O2}Z_{G2}(1+Z_{O2}+Z_{O2}Z_{B2})}\\[.5em]
       -\frac{1+Z_{O2}+Z_{O2}Z_{B2}}{Z_{B2}(1+Z_{G2}+Z_{O2}Z_{G2})} & 1+\frac{e^{2\pi i \theta_1}}{Z_{O2}Z_{B2}Z_{G2}}
    \end{pmatrix},\\[.5em]
    \widetilde{M}_2=&\begin{pmatrix}
        1+\frac{e^{2\pi i \theta_2}(1+Z_{B2})(1+Z_{G2}+Z_{O2}Z_{G2})}{Z_{O2}Z_{B2}Z_{G2}(1+Z_{B2}+Z_{B2}Z_{G2})} & \frac{(1+Z_{B2})(1+Z_{G2}+Z_{O2}Z_{G2})}{Z_{G2}}\left(\frac{e^{2\pi i \theta_2}}{Z_{O2}Z_{B2}(1+Z_{B2}+Z_{B2}Z_{G2})}+\frac{1}{1+Z_{O2}+Z_{O2}Z_{B2}}\right)\\[.5em]
        -\frac{e^{2\pi i \theta_2}(1+Z_{O2}+Z_{O2}Z_{B2})}{Z_{O2}Z_{B2}(1+Z_{B2}+Z_{B2}Z_{G2})} & -\frac{e^{2\pi i \theta_2}(1+Z_{O2}+Z_{O2}Z_{B2})}{Z_{O2}Z_{B2}(1+Z_{B2}+Z_{B2}Z_{G2})}
    \end{pmatrix},\\[.5em]
    \widetilde{M}_3=&\begin{pmatrix}
        \frac{e^{2\pi i \theta_3}}{Z_{O2}Z_{B2}Z_{G2}}+\frac{(1+Z_{G2})(1+Z_{O2}+Z_{O2}Z_{B2})}{Z_{O2}(1+Z_{B2}+Z_{B2}Z_{G2})} & \frac{1+Z_{G2}+Z_{O2}Z_{G2}}{Z_{O2}(1+Z_{B2}+Z_{B2}Z_{G2})}\\[.5em]
        -\frac{(1+Z_{G2})(1+Z_{O2}+Z_{O2}Z_{B2})\left(e^{2\pi i \theta_3}(1+Z_{B2}+Z_{B2}Z_{G2})+Z_{B2}Z_{G2}(1+Z_{G2}+Z_{O2}Z_{G2})\right)}{Z_{O2}Z_{B2}Z_{G2}(1+Z_{B2}+Z_{B2}Z_{G2})(1+Z_{G2}+Z_{O2}Z_{G2})} & -\frac{1+Z_{G2}+Z_{O2}Z_{G2}}{Z_{O2}(1+Z_{B2}+Z_{B2}Z_{G2})}
    \end{pmatrix}.
\end{aligned}$}
\end{equation}
The notation sticks with $M$ instead of $N$ to stress the dimensional invariance. Notice that the new triple is now defined over the function field of $\pazocal{X}_\triangledown$.  

\begin{remark}
    Normalizing back to $SL_2(\mathbb{C})$, direct computations prove the transformed triple \eqref{tildeN} is indeed the one corresponding to $w_2$: its local monodromy match the multiplicative version \eqref{multparamchange} of $w_2(\theta)$, while global monodromy data are preserved:
\begin{equation}
    \tilde{x}_i=\iota_\infty\iota_i\mathrm{tr}(\widetilde{M}_j\widetilde{M}_k)=\mathrm{tr}(M_jM_k)=x_i,
\end{equation}
for distinct $i,j,k\in\{1,2,3\}$. 
\end{remark}

We first focus on the spectra, as data independent of the basis completion:
\begin{gather*}
    \mathrm{eig}(\widetilde{M}_i)=\left\{1,\frac{e^{2\pi i\theta_i}}{Z_{O2}Z_{B2}Z_{G2}}\right\},\\[.25em]
    \widetilde{M}_\infty:=(\widetilde{M}_1\widetilde{M}_2\widetilde{M}_3)^{-1}=\begin{pmatrix}
        e^{-2\pi i(\theta_1+\theta_2+\theta_3)}Z_{O2}^2Z_{B2}^2Z_{G2}^2 & 0\\
        -\tilde{z} & Z_{O2}Z_{B2}Z_{G2}
    \end{pmatrix}
\end{gather*}
for 
\begin{equation*}\resizebox{\textwidth}{!}{$\begin{aligned}
    \tilde{z}=&\,Z_{O2}Z_{G2}(1+Z_{O2}+Z_{O2}Z_{B2})\\\cdot&\frac{1+e^{-2\pi i\theta_1}Z_{O2}(1+Z_{B2}+Z_{B2}Z_{G2})+Z_{O2}Z_{B2}\big(e^{-2\pi i(\theta_1+\theta_2+\theta_3)}Z_{B2}Z_{G2}(1+Z_{G2})+e^{-2\pi i(\theta_1+\theta_2)}(1+Z_{B2}+Z_{B2}Z_{G2})\big)}{1+Z_{G2}+Z_{O2}Z_{G2}}.
    \end{aligned}$}
\end{equation*}
These formulae imply our middle convolution entails the change of parameters
\begin{equation}\label{multparamchange}\resizebox{\textwidth}{!}{$
    \left(e^{2\pi i\theta_1},e^{2\pi i\theta_2},e^{2\pi i\theta_3},\nu,\nu^{-1}e^{-2\pi i (\theta_1+\theta_2+\theta_3)}\right) \longmapsto \left(\nu e^{2\pi i\theta_1},\nu e^{2\pi i\theta_2},\nu e^{2\pi i\theta_3},\nu^{-1},\nu^{-2}e^{-2\pi i (\theta_1+\theta_2+\theta_3)}\right),$}
\end{equation}
which is exactly the multiplicative analogue of $w_2(\theta)$, cf. \cite[(5.18)]{Filipuk2007}. The latter is recovered by lifting \eqref{multparamchange} directly to the system's parameters as the cyclic formulae
\begin{equation}\label{thetachange}
w_2(\theta_1,\theta_2,\theta_3,\theta_\infty)=\textstyle\left(\frac{+\theta_1-\theta_2-\theta_3-\theta_\infty}{2},\frac{-\theta_1+\theta_2-\theta_3-\theta_\infty}{2},\frac{-\theta_1-\theta_2+\theta_3-\theta_\infty}{2},\frac{-\theta_1-\theta_2-\theta_3+\theta_\infty}{2}\right).
\end{equation}

\subsection{Cluster features}\label{sec:combinatorics}

We now focus on the specific shape of our triple \eqref{tildeN}, which is the very outcome of the chosen basic completion.

Notice that each matrix $\widetilde{M}_i$ reproduces the same pattern of the corresponding $\widehat{M}_i$: e.g., for $i=1$ both are lower antitriangular with $(2,2)$-entries corresponding precisely via change \eqref{multparamchange}. In fact, completion \eqref{abcd} is the unique one allowing to encapsulate our Okamoto-type multiplicative middle convolution as a single transformation of the \emph{whole} set of coordinate.
Explicitly, \eqref{multparamchange} reads on the Casimir $\pazocal{X}$-coordinates as
\begin{equation}\label{Caschange}
    \begin{aligned}
    &Z_{O1}=\iota_1^{-2} \longmapsto \iota_1^{-1}\iota_2\iota_3\iota_\infty,\\
    &Z_{B1}=\iota_2^{-2} \longmapsto \iota_1\iota_2^{-1}\iota_3\iota_\infty,\\
    &Z_{G1}=\iota_3^{-2} \longmapsto \iota_1\iota_2\iota_3^{-1}\iota_\infty,
    \end{aligned}
\end{equation}
and is extended to the rest of the chart as
\begin{equation}\label{Zchange}
\begin{aligned}
    Z_{O2} \longmapsto \tilde{Z}_{O2}:=\frac{1+Z_{O2}+Z_{O2}Z_{B2}}{Z_{B2}(1+Z_{G2}+Z_{O2}Z_{G2})},\\[.5em]
    Z_{B2} \longmapsto \tilde{Z}_{B2}:=\frac{1+Z_{B2}+Z_{B2}Z_{G2}}{Z_{G2}(1+Z_{O2}+Z_{O2}Z_{B2})},\\[.5em]
    Z_{G2} \longmapsto \tilde{Z}_{G2}:=\frac{1+Z_{G2}+Z_{O2}Z_{G2}}{Z_{O2}(1+Z_{B2}+Z_{B2}Z_{G2})}.
\end{aligned}
\end{equation}

Cyclic formulae \eqref{Zchange}, in the anticipated form \eqref{tastemut}, are deduced by just comparing entries in the same-shaped triples $\widehat{\mathbf{M}}$ and $\widetilde{\mathbf{M}}$; it is straightforward to check that they (entry-wise) transform the former triple into the latter.
Moreover, they precisely entail the inversion of $\nu$ prescribed by \eqref{multparamchange}, since
\begin{equation}
    \tilde{Z}_{O2}\tilde{Z}_{B2}\tilde{Z}_{G2}=\frac{1}{Z_{O2}Z_{B2}Z_{G2}}.
\end{equation}

\subsubsection{Ensemble viewpoint}

The cluster Poisson properties underlying the rational triple of assignments \eqref{Zchange} are best understood in the language of ensembles, which is given a minimal description in \Cref{app:X&A}.

From its native $\pazocal{X}$-viewpoint, the triple is naturally expressed as the cluster transformation
\begin{equation}\label{clustrans}
    \sigma_{O}\circ(\mu_O\circ\mu_G\circ\mu_B\circ\mu_O),
\end{equation}
where $\sigma_{O}$ denotes the quiver isomorphism permuting vertices $B$ and $G$. \Cref{fig:clustrans} gives a step-by-step visualization of the cluster transformation at the level of quivers, while \Cref{tab:Zmutated} gives formulae for each mutated chart over the initial coordinates. Notice that, in quiver terms, the cluster transformation simplifies to the identity map: this echoes the defining pattern-preservation property of the map \eqref{Zchange} itself.
\begin{figure}
    \centering
    \includegraphics[width=.8\textwidth]{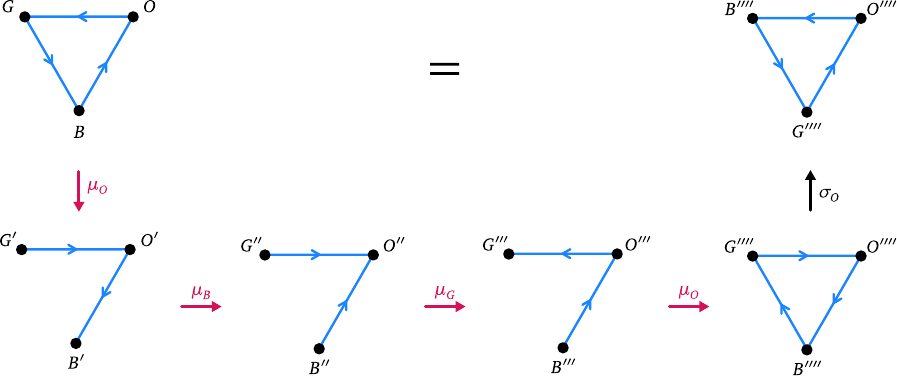}
    \caption{Cluster transformation \eqref{clustrans} applied to the nontrivial component of the quiver $\triangledown$, whose vertex set $\{O,B,G\}$ is labeled after the corresponding (subscript $2$) $\pazocal{X}$-coordinates.}\label{fig:clustrans}
\end{figure}
\begin{table}\centering
\begin{tabular}{llcccl}\toprule
$\pazocal{Q}$ & $\hspace{3em}$ & \multicolumn{3}{c}{$\pazocal{X}$-coordinates}\\\cmidrule(lr){3-5}\midrule
$\triangledown$   & $\hspace{3em}$ & $Z_{O2}$ & $Z_{B2}$ & $Z_{G2}$ \\[.75em]
$\triangledown'$ & $\hspace{3em}$ & $\frac{1}{Z_{O2}}$ & $\frac{Z_{O2}Z_{B2}}{1+Z_{O2}}$ & $Z_{G2}(1+Z_{O2})$  \\[.75em]
$\triangledown''$ & $\hspace{3em}$ & $\frac{Z_{B2}}{1+Z_{O2}+Z_{O2}Z_{B2}}$ & $\frac{1+Z_{O2}}{Z_{O2}Z_{B2}}$ & $Z_{G2}(1+Z_{O2})$ \\[.75em]
$\triangledown'''$  & $\hspace{3em}$ & $\frac{Z_{B2}(1+Z_{G2}+Z_{O2}Z_{G2})}{1+Z_{O2}+Z_{O2}Z_{B2}}$ & $\frac{1+Z_{O2}}{Z_{O2}Z_{B2}}$ & $\frac{1}{Z_{G2}(1+Z_{O2})}$  \\[.75em]
$\triangledown''''$  & $\hspace{3em}$ & $\frac{1+Z_{O2}+Z_{O2}Z_{B2}}{Z_{B2}(1+Z_{G2}+Z_{O2}Z_{G2})}$ & $\frac{1+Z_{G2}+Z_{O2}Z_{G2}}{Z_{O2}(1+Z_{B2}+Z_{B2}Z_{G2})}$ & $\frac{1+Z_{B2}+Z_{B2}Z_{G2}}{Z_{G2}(1+Z_{O2}+Z_{O2}Z_{B2})}$ \\\bottomrule
\end{tabular}\caption{Cluster $\pazocal{X}$-charts for the sequence of mutations in the cluster transformation \eqref{clustrans}. The first column details the quiver from \Cref{fig:clustrans} corresponding to each chart, whose cluster $\pazocal{X}$-coordinates are written over the initial ones as prescribed by formulae \eqref{Xmutation}.}\label{tab:Zmutated}
\end{table}
\begin{remark}\label{rmk:reflection}
    Cluster transformations \eqref{clustrans} are well-known: e.g., as the $\mathrm{DT}$-transformation $\mathbf{K}$ of $A_3$ \cite{GS2018} or the reflection $R_{A_3}$ \cite{Masuda2024}. In particular, one can turn the permutation $\sigma$ into a sequence of five mutations using the so-called ``pentagon relation'', matching the anticipated form \eqref{mutprel}. 
\end{remark}

Seen as a cluster transformation, the cyclic triple upgrades to a Poisson isomorphism
\begin{equation}
        \mathbb{C}_{\tilde{\iota}}(\pazocal{X}_\triangledown)\, \xrightarrow{\raisebox{-.75ex}[0ex][0ex]{\,\,$\sim$\,\,}}\,\mathbb{C}_{\iota}(\pazocal{X}_\triangledown)
\end{equation} 
where
\begin{equation}
\mathbb{C}_{\tilde{\iota}}(\pazocal{X}_\triangledown):=\mathbb{C}\big( \tilde{Z}_{O2},\tilde{Z}_{B2},\tilde{Z}_{G2}\big)\big/(\tilde{Z}_{O2}\tilde{Z}_{B2}\tilde{Z}_{G2}-\tilde{\iota}_1\tilde{\iota}_2\tilde{\iota}_3\tilde{\iota}_\infty)
\end{equation}
for 
\begin{equation}
    \tilde{\iota}_l^2=\frac{\iota_l}{\iota_i\iota_j\iota_k},\quad i\neq j\neq k\neq l\in\{1,2,3,\infty\}.
\end{equation}
Mirroring the quiver picture, the ideal remains cubic.

Through the homomorphism $p$ \eqref{p}, we can frame the cluster transformation in terms of the whole ensemble.
In particular, denoting by $A_k$ the cluster $\pazocal{A}$-coordinate of $\pazocal{A}_\triangledown$ at vertex $k$,
\begin{equation}\label{ZtoC}
    (Z_{O1}, Z_{B1}, Z_{G1},Z_{O2}, Z_{B2}, Z_{G2})\xmapsto{\hspace{.5em}p^*\hspace{.25em}}
\left(1,1,1,\frac{A_{G2}}{A_{B2}},\frac{A_{O2}}{A_{G2}},\frac{A_{B2}}{A_{O2}}\right).
\end{equation}
As expected, this is a projection to the (unital level) symplectic leaf of $\pazocal{X}_\triangledown$: the first triple, made of Casimir elements, gets evaluated.
It is then straightforward to check that \eqref{Zchange} simplifies to the identity map when pull-backed to the $\pazocal{A}$-space.
Since $p$ commutes with mutations, this simplifications can be equivalently deduced from the action of the cluster transformation on $\pazocal{A}$-coordinates: e.g., from \Cref{tab:Amutated} we get that
\begin{equation}
    \frac{A_{G2}}{A_{B2}}\xmapsto{\hspace{.5em}\mu_O\circ\mu_G\circ\mu_B\circ\mu_O\hspace{.5em}}\frac{\phantom{.}\frac{A_{O2}+A_{B2}+A_{G2}}{A_{O2}A_{G2}}\phantom{.}}{\frac{A_{O2}+A_{B2}+A_{G2}}{A_{O2}A_{B2}}}\xmapsto{\hspace{.5em}\sigma_{O}\hspace{.5em}}\frac{\phantom{.}\frac{A_{O2}+A_{B2}+A_{G2}}{A_{O2}A_{B2}}\phantom{.}}{\frac{A_{O2}+A_{B2}+A_{G2}}{A_{O2}A_{G2}}}=\frac{A_{G2}}{A_{B2}}.
\end{equation}
In other words, the Poisson map realizing $w_2$ is invisible to the subtorus $\pazocal{U}_\triangledown=p(\pazocal{A}_\triangledown)$ and thus to be understood as a pure change of the Casimirs' level set.

\begin{figure}[!t]
    \centering
    \includegraphics[width=.8\textwidth]{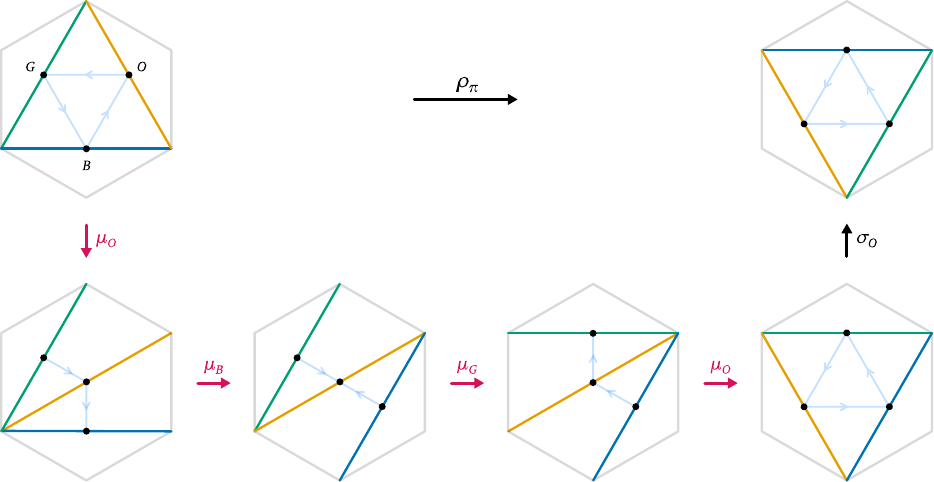}
    \caption{Cluster transformation \eqref{clustrans} as flips of colored triangulations in the hexagon $\{6\}$. On equilateral triangulations, the whole map amounts to the rotation $\rho_{\pi}$.}\label{fig:triangulations}
\end{figure}

\begin{table}[!t]\centering
\begin{tabular}{llcccl}\toprule
$\pazocal{Q}$ & $\hspace{3em}$  & \multicolumn{3}{c}{$\pazocal{A}$-coordinates}\\\cmidrule(lr){3-5}\midrule
$\triangledown$   & $\hspace{3em}$  & $A_{O2}$ & $A_{B2}$ & $A_{G2}$ \\[.5em]
$\triangledown'$ & $\hspace{3em}$  & $\frac{A_{B2}+A_{G2}}{A_{O2}}$ & $A_{B2}$ & $A_{G2}$ \\[.5em]
$\triangledown''$ & $\hspace{3em}$  & $\frac{A_{B2}+A_{G2}}{A_{O2}}$ & $\frac{A_{O2}+A_{B2}+A_{G2}}{A_{O2}A_{B2}}$ & $A_{G2}$ \\[.5em]
$\triangledown'''$  & $\hspace{3em}$  & $\frac{A_{B2}+A_{G2}}{A_{O2}}$ & $\frac{A_{O2}+A_{B2}+A_{G2}}{A_{O2}A_{B2}}$ & $\frac{A_{O2}+A_{B2}+A_{G2}}{A_{O2}A_{G2}}$ \\[.5em]
$\triangledown''''$  & $\hspace{3em}$ & $\frac{A_{O2}+A_{B2}+A_{G2}}{A_{B2}A_{G2}}$ & $\frac{A_{O2}+A_{B2}+A_{G2}}{A_{O2}A_{B2}}$ & $\frac{A_{O2}+A_{B2}+A_{G2}}{A_{O2}A_{G2}}$ \\\bottomrule
\end{tabular}
\caption{Cluster $\pazocal{A}$-charts for the sequence of mutations in the cluster transformation \eqref{clustrans}. The first column details the quiver from \Cref{fig:clustrans} corresponding to each chart, whose cluster $\pazocal{A}$-coordinates are written over the initial ones as prescribed by formulae \eqref{Amutation}.}\label{tab:Amutated}
\end{table}

\subsubsection{Associahedron viewpoint}\label{sec:hedron}

As it happens, the combinatorial nature of the cyclic triple \eqref{Zchange} goes far beyond the specific cluster formula \eqref{clustrans}. Indeed, that sequence of mutations is not unique and a whole equivalent family exists, codified by the dynamics of flip operations.  
In particular, the most essential description for this family reads in terms of the combinatorics of flips of triangulations on the regular hexagon $\{6\}$---as six are the sides resulting from gluing four triangles (cf. \Cref{fig:2dimloops}).

A triangulation on $\{6\}$ is singled out by a triple of non-intersecting chords, and a flip removes a selected chord to replace it with the other diagonal in the resulting quadrilateral. The overall combinatorics is captured by the associahedron $\mathbb{A}_3$, whose vertices correspond to triangulations and edges represent flips \cite[Figure 1]{Fomin2021}.
In our dictionary, each triangulation corresponds to a cluster $\{Z_{O2},Z_{B2},Z_{G2}\}$ and the $\pazocal{X}$-mutation $\mu_\alpha$, $\alpha=O,B,G$, flips the chord corresponding to $Z_{\alpha2}$. In order to fully encode mutation \emph{formulae}, we need to tell apart the dynamics of each individual coordinate: this further layer of detail can be captured by coloring the chords of the triangulation, i.e., by passing to the \textbf{colorful associahedron} $\mathbb{A}^c_3$.

This abstract polytope is known \cite{Pardo2015} to tessellate an orientable genus $4$ surface with $18$ decagonal and $18$ tetragonal faces totaling $84$ three-valent vertices, one for each colored triangulation of the hexagon, of which $12$ are surrounded by three decagons and $72$ by two decagons and one tetragon. In particular, its $1$-skeleton is a connected $3$-regular graph and $\mathbb{A}^c_3/\mathfrak{S}_3\simeq\mathbb{A}_3$, for the symmetric group acting by color permutations.

\Cref{fig:triangulations} illustrates the dynamics on $\mathbb{A}^c_3$ induced by formula \eqref{clustrans}. The starting cluster, singled out by the Teichm\"uller parametrization, corresponds to the top-left equilateral triangulation, thereafter referred to as the \textbf{reference} one, whose chords have been colored after the coordinates' subscripts.
Notice that the final triangulation, itself equilateral, is obtained by rotation of $\pi$.
There is a total of twelve such equilateral triangulations, halved by upward or downward orientation: these are precisely the vertices of $\mathbb{A}^c_3$ belonging only to decagons.

We are about to prove that $w_2$, as a uniform operation requiring no defining vertex, reads as the $\pi$-rotation between the unique upward and downward equilateral triangulations of the standard associahedron $\mathbb{A}_3$, whose colorful covering $\mathbb{A}^c_3$ becomes necessary when passing to the language of mutations. Remarkably, the geometry of $\mathbb{A}^c_3$ is the natural structure encapsulating the many mutational expressions of $w_2$: 
\begin{theorem}\label{thm:main}
    The mutation formula attached to a sequence of flips of colored triangulations on $\mathbb{A}^c_3$ is path-independent, i.e., its mapping of clusters is uniquely determined by the initial and final triangulations the sequence connects between. In particular, rotation of $\pi$ admits a well-defined involutive mutation formula for equilateral triangulations on $\mathbb{A}^c_3$, whose explicit map of clusters reads on the reference triangulation as \eqref{Zchange}. 
\end{theorem}
\begin{proof}
Being mutations involutive and $\mathbb{A}^c_3$ connected, we can restrict to paths between equilateral triangulations which, up to $\pi$-rotation, differ at most by two chord-permutations.
Moreover, each tetragon can be treated as a single $4$-valent vertex: it is easily checked that the two paths between an ordered pair of its vertices coincide as maps on clusters.  
This follows from the fact that triangulations labeling a tetragon share a longer colored chord (cf. \Cref{fig:colorass}), and the two flips of the shorter ones commute also as cluster mutations.

Then, the reduced building blocks connecting between any two equilateral triangulations are sequences either alternated $\mu_\beta\mu_\gamma\mu_\beta\mu_\gamma\mu_\beta$ or in form $\mu_\alpha\mu_\gamma\mu_\beta\mu_\alpha=\mu_\alpha\mu_\beta\mu_\gamma\mu_\alpha$ (cf. \Cref{fig:nearest}). In geometric terms, these two sequences stem from decagons and tetragons, respectively.
 
Decagons arise by fixing a shorter chord in a triangulation's dynamics: mutations must flip between $2$-colored triangulations of a \emph{pentagon}, whose evolution is precisely the $10$-cycle $\mathbb{A}^c_2$---which is a $2$-fold covering of $\mathbb{A}_2$, see \cite[Figure 2]{Pardo2015}. Motion along the decagon fixing chord $\alpha$ corresponds to alternating mutations $\mu_\beta$ and $\mu_\gamma$, for $\beta,\gamma\neq\alpha$, whose full $10$-cycle simplifies to the identity, i.e., the Coxeter-type relation
\begin{equation}
    (\mu_\gamma\mu_\beta)^5=\mathrm{id} 
\end{equation}
holds. Faithful to the triangulation dynamics, each half decagon connecting between its equilateral triangulations of $\mathbb{A}^c_3$ delivers a mutation formula for permutations of $\pazocal{X}$-coordinates:
\begin{equation}\label{sigma}
    \sigma_{\alpha}=(\mu_\beta\mu_\gamma)^2\mu_\beta=(\mu_\gamma\mu_\beta)^2\mu_\gamma. 
\end{equation}
In particular,
\begin{equation*}
    \begin{matrix}
        \sigma^*_{\alpha} & : & \{Z'_{\alpha},Z'_{\beta},Z'_{\gamma}\} & \longrightarrow & \{Z_{\alpha},Z_{\beta},Z_{\gamma}\}\\[.25em]
         & & Z'_\alpha & \longmapsto & Z_{\alpha},\\
         & & Z'_\beta & \longmapsto & Z_{\gamma},\\
         & & Z'_\gamma & \longmapsto & Z_{\beta}
    \end{matrix}
\end{equation*}
as detailed in \Cref{tab:sigmamut} for the $\alpha=O$ case.
Therefore, for a fixed upward or downward orientation, formula \eqref{sigma} identifies the dynamics of equilateral triangulations with that induced on clusters: backward comparison of chords' colors coincides with the mapping due to the sequence of cluster mutations, making the latter path-independent.
\begin{table}[!t]\centering
\begin{tabular}{llcccl}\toprule
Mutation steps of $\sigma_{O}$ & $\hspace{3em}$ & \multicolumn{3}{c}{$\pazocal{X}$-coordinates}\\\cmidrule(lr){3-5}\midrule
\multicolumn{1}{c}{$\mu_{B}$}   & $\hspace{3em}$ & $(1+Z_{B2})Z_{O2}$ & $Z_{B2}^{-1}$ & $\frac{Z_{B2}Z_{G2}}{1+Z_{B2}}$ \\[.75em]
\multicolumn{1}{c}{$\mu_{G}$} & $\hspace{3em}$ & $(1+Z_{B2})Z_{O2}$ & $\frac{Z_{G2}}{1+Z_{B2}+Z_{B2}Z_{G2}}$ & $\frac{1+Z_{B2}}{Z_{B2}Z_{G2}}$  \\[.75em]
\multicolumn{1}{c}{$\mu_{B}$} & $\hspace{3em}$ & $\frac{Z_{O2}Z_{G2}}{1+Z_{G2}}$ & $\frac{1+Z_{B2}+Z_{B2}Z_{G2}}{Z_{G2}}$ & $\frac{1}{Z_{B2}(1+Z_{G2})}$ \\[.75em]
\multicolumn{1}{c}{$\mu_{G}$}  & $\hspace{3em}$ & $\frac{Z_{O2}Z_{G2}}{1+Z_{G2}}$ & $Z_{G2}^{-1}$ & $Z_{B2}(1+Z_{G2})$  \\[.75em]
\multicolumn{1}{c}{$\mu_{B}$}  & $\hspace{3em}$ & $Z_{O2}$ & $Z_{G2}$ & $Z_{B2}$ \\\bottomrule
\end{tabular}\caption{Sequence of mutations $\sigma_{O}$ on the reference triangulation, with cluster $\pazocal{X}$-charts written over the initial one $\{Z_{O2},Z_{B2},Z_{G2}\}$ as prescribed by formulae \eqref{Xmutation}. In terms of triangular quivers, notice that $\sigma_{O}$ entails a reversal of orientation, from the reference counterclockwise to clockwise.}\label{tab:sigmamut}
\end{table}

We are left to deal with orientation reversals. As anticipated, this is where $w_2$ manifests as the path-independent mutation formula for rotation of $\pi$. Following \Cref{fig:colorass}, each {\color{2red}red} $4$-path going through a tetragon is an equivalent $w_2$ dynamics: up to a final permutation, the cluster transforms as \eqref{Zchange}. In other words, for $\beta,\gamma\neq\alpha$,
\begin{equation}\label{w2mutation}
    \mu_{w_2}:=(\mu_\beta\mu_\gamma\mu_\beta\mu_\gamma\mu_\beta)(\mu_\alpha\mu_\gamma\mu_\beta\mu_\alpha)=(\mu_\gamma\mu_\beta\mu_\gamma\mu_\beta\mu_\gamma)(\mu_\alpha\mu_\gamma\mu_\beta\mu_\alpha).
\end{equation}
Crucially, this sequence of mutations is involutive and commutes with any permutation, behaving precisely as half of a full rotation. Since permutations reverse the quiver orientation, notice the commuting property requires the explicit formula for sequence \eqref{w2mutation} on the \emph{clockwise} triangular quiver, which we include for completeness:
\begin{equation}
\begin{aligned}
    \mu_{w_2}^*(Z'_{O2})&=\frac{1+Z_{O2}+Z_{O2}Z_{G2}}{Z_{G2}(1+Z_{B2}+Z_{O2}Z_{B2})},\\
    \mu_{w_2}^*(Z'_{B2})&=\frac{1+Z_{B2}+Z_{O2}Z_{B2}}{Z_{O2}(1+Z_{G2}+Z_{B2}Z_{G2})},\\
    \mu_{w_2}^*(Z'_{G2})&=\frac{1+Z_{G2}+Z_{B2}Z_{G2}}{Z_{B2}(1+Z_{O2}+Z_{O2}Z_{G2})}.
\end{aligned}
\end{equation}
 It follows that the whole dynamics of equilateral triangulations is identified with its mutation analogue on clusters: whenever a path connects between oppositely oriented equilateral triangulations, one can always assume a $\pi$-rotation as first operation and deduce the overall mutation by backward color comparison---i.e., a path-independent recipe.
\end{proof}
\begin{figure}[!t]
    \centering
    \includegraphics[width=.315\textwidth]{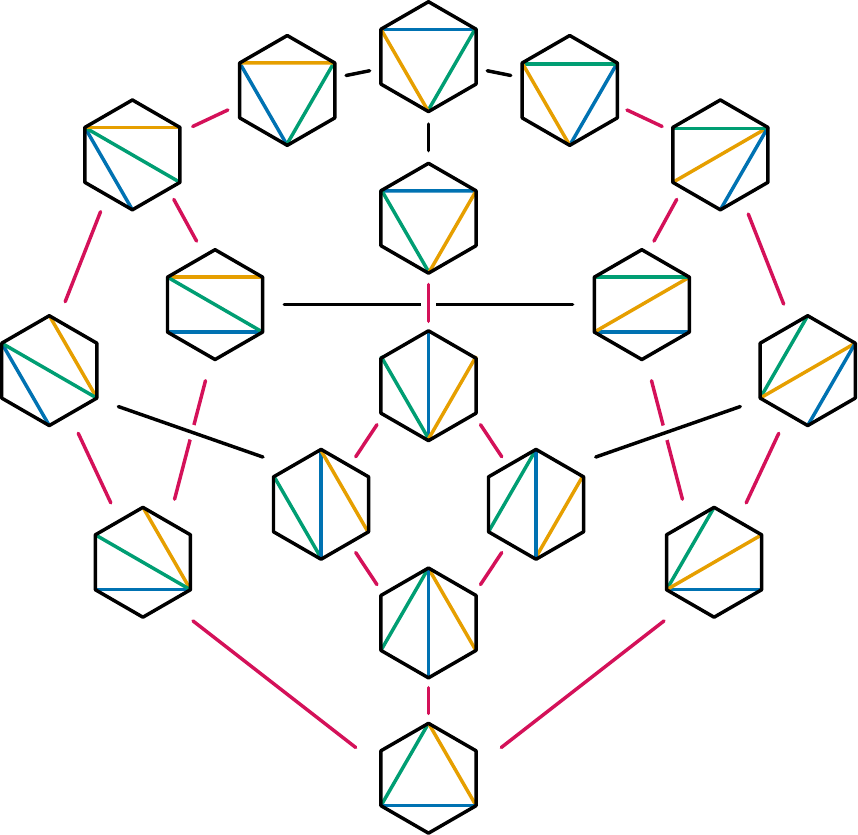}
    \caption{Shortest paths on the colored $3$-associahedron amounting to a $\pi$-rotation of the reference triangulation. Each upward {\color{2red}red} $4$-path encapsulates $w_2$, up to a final permutation indicated in black. The other triple of black edges delimitate the unique three decagons incident to the reference triangulation. The rightmost shortest path is the one in \Cref{fig:triangulations}.}\label{fig:colorass}
\end{figure}
\begin{figure}[!t]
    \centering
    \includegraphics[width=.3\textwidth]{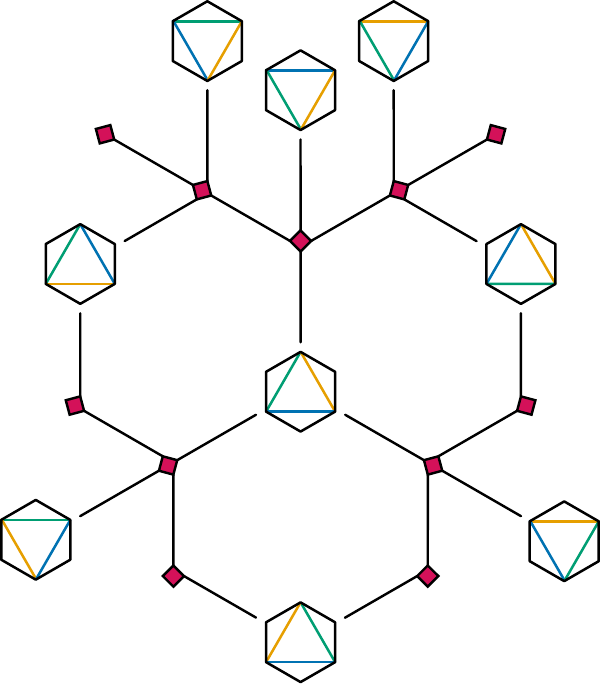}
    \caption{All distance $2$ equilateral triangulations obtained via chord-reflections of our reference triangulation, together with all distance $3$ ones obtained either by minimal rotation or a single chord-permutation. Tetragons are simplified to square-shaped {\color{2red}red} vertices. All three incident edges are displayed for the reference triangulation, whose incident decagons read in this schematic as the three adjacent hexagons.}\label{fig:nearest}
    \end{figure}
\begin{remark}
    Paths move from an equilateral triangulation to another, hitting one or two tetragons at each step (steps can hit up to six tetragons, but are factorized over these two minimal cases).
    The former case entails the reflection along a chord $\alpha$, and is therefore expressed as $\sigma_{\!\alpha}\mu_{w_2}$.
    The latter is twofold, depending on whether the tetragons belong to the same decagon, resulting in $\sigma_{\!\alpha}$ for the fixed chord $\alpha$, or two adjacent ones, resulting in a $\pm\frac{\pi}{3}$-rotation written as $\sigma_{\!\gamma}\sigma_{\!\beta}\mu_{w_2}$ for a suitable pair of permutations. See \Cref{fig:nearest} for examples of all three cases.
    In particular, permutations $\sigma_{\!\alpha}$ give the symmetry group of the equilateral triangle, namely the Weyl group $W(A_2)$. Their compositions amount to rotations of $\frac{2\pi n}{3}$, $n\in\mathbb{Z}$, that cannot achieve $\mu_{w_2}$. This fact is manifest in the language of triangulations which, unlike the quiver-theoretic one, is able to detect---and distinguish---all rotations.
\end{remark}

\begin{remark}\label{rmk:qPainlevé}
    As already observed, each permutation reverses the orientation of the reference quiver $\triangledown$. On the contrary, the quiver is invariant under $\mu_{w_2}$. Such invariance plays a pivotal role in the $q$-difference Painlevé world, which fits the discrete dynamics naturally associated to cluster mutations, so much so that all $q$-Painlevé equations generate from mutation-periodic quivers \cite{BGM2018}. Moreover, the equation's symmetry group itself can be realized in terms of mutations and permutations. The setting here developed for ${w_2}$ invites to study similar phenomena also in the continuous differential world.  
\end{remark}

\subsubsection{Fat graph viewpoint}

A more involved dual description of mutation formula \eqref{w2mutation} for the monodromic $w_2$ can be given as a special reversal operation for star-shaped fat graphs. Despite requiring to further keep track of how the Casimir $\pazocal{X}$-coordinates change along with geometry, this dual viewpoint results in a totally equivalent cluster dynamics.

A triangulation of the hexagon $\{6\}$ dualizes to a \textbf{spine} $\Gamma_{0,4}$, i.e., a $3$-valent fat graph living on the $4$-punctured Riemann sphere without self-intersection, each face of which contains exactly one puncture. \Cref{fig:2dimloops} provides an example of $\Gamma_{0,4}$ with the case of our reference triangulation.
Mutation $\mu_\alpha$ flips the dog bone component of the fat graph dual to chord $\alpha$ as displayed in \Cref{fig:fatgraphflip}. In the quiver terms prescribed by the Teichm\"uller theory for $PSL_2$, each triangle entails a counterclockwise triangular quiver whose variables ``amalgamate'' at adjacent sides with formula $Z_\alpha:=Z_aZ_b$ \cite{Goncharov2022}, for $Z_a$ and $Z_b$ denoting the variables attached to the pair of sides glued into chord $\alpha$.

On the one hand, each variable labeling a chord generates via \textbf{amalgamation}, and the corresponding quiver mutation matches the fat-graph flip.
On the other hand, each Casimir corresponds to a loop which, in $\pazocal{X}$-coordinates, results as product of the variables labeling all the sides it crosses.
In particular, when the end of the dog bone dual to chord $\alpha$ self-glues to a loop, namely the dual triangle identifies two of its own sides, the flip adds to the Casimir element a factor $Z'_{\alpha}$ on top of mutation formulae.
\begin{figure}[t]
    \centering
    \includegraphics[width=.7\textwidth]{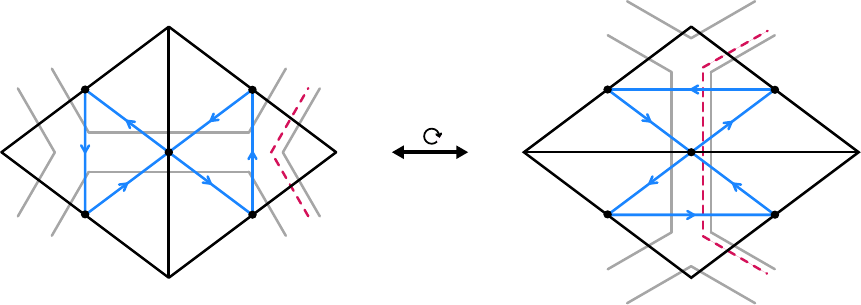}
    \caption{Flip operation for fat graphs as a rotation of a dog bone component. When the {\color{2red}red} path is a portion of a loop, the corresponding Casimir element acquires/loses a factor of the dogbone's middle variable (respectively, in the rightward/leftward case).}\label{fig:fatgraphflip}
\end{figure}
\begin{example}\label{ex:1}
    Let the right end of the dog bone in \Cref{fig:fatgraphflip} (left) be self-glued and, denoting by $\alpha$ the dog bone's dual cord, label counterclockwise by $Z_\alpha,Z_b,Z_t$ the $\pazocal{X}$-coordinates attached to the right triangle. Then, the Casimir element for the resulting {\color{2red}red} loop reads as $Z_{R1}:=Z_{t}Z_{b}$. After the flip, the loop crosses also chord $\alpha$ and the Casimir updates to $Z'_{R1}=Z'_{t}Z'_{\alpha}Z'_{b}$.
\end{example}
Following the above recipe, the sequence of mutations $\mu_{w_2}$ \eqref{w2mutation} reverses the star-shape of the initial fat graph giving rise to an operation we name \textbf{inside-out}. A visual description for the case $\alpha=O,$ $\beta=B,$ $\gamma=G$ is given in \Cref{fig:insideout}.
It is a direct check that Casimir values are preserved under each mutation step of the inside-out, in accordance with the simpler description of \Cref{sec:hedron} just in terms of chords.
\begin{figure}
    \centering
    \includegraphics[width=.6\textwidth]{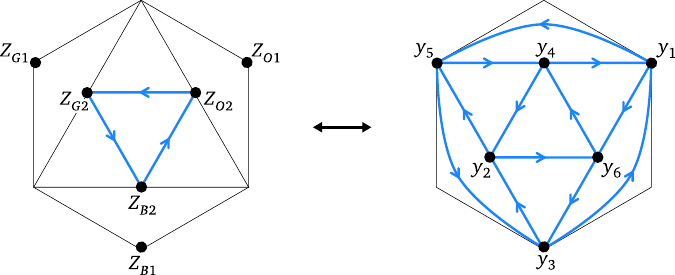}
    \caption{Quivers attached to the self-folded triangulation (left) and the non self-folded one (right). No sequence of standard $\pazocal{X}$-mutations can map between the two, and more general Poisson isomorphisms, like the one defined by \eqref{Poissoniso}, are needed for their identification.}\label{fig:amalquivers}
\end{figure}
\begin{example}
    \Cref{tab:ZO1} tracks the Casimir element $Z_{O1}$, which is generated by the {\color{3ochre}ochre} loop of the reference triangulation in \Cref{fig:insideout}, in its dynamics under the chord-reflection component of the inside-out. Written in terms of the initial cluster $\{Z_{O1},Z_{O2},Z_{B1},Z_{B2},Z_{G1},Z_{G2}\}$, whose subscript $1$ (subscript $2$) coordinates correspond to colored loops (colored chords), each step simplifies to the identity map.
\end{example}

\begin{table}\centering
\begin{tabular}{cc}\toprule
Mutation & Casimir element ({\color{3ochre}ochre})\\\midrule
$\mu_{O}$   & $Z'_{O1}Z'_{O2}=(Z_{O1}Z_{O2})Z_{O2}^{-1}$ \\[.75em]
$\mu_{B}$ & $Z''_{O1}Z''_{O2}Z''_{B2}=\big(Z_{O1}Z_{O2}\frac{1+Z_{O2}+Z_{O2}Z_{B2}}{1+Z_{O2}}\big)\frac{Z_{B2}}{1+Z_{O2}+Z_{O2}Z_{B2}}\frac{1+Z_{O2}}{Z_{O2}Z_{B2}}$  \\[.75em]
$\mu_{G}$ & $Z'''_{O1}Z'''_{O2}Z'''_{B2}Z'''_{G2}=\big(Z_{O1}Z_{O2}Z_{G2}\frac{1+Z_{O2}+Z_{O2}Z_{B2}}{1+Z_{G2}+Z_{O2}Z_{G2}}\big)\big(Z_{B2}\frac{1+Z_{G2}+Z_{O2}Z_{G2}}{1+Z_{O2}+Z_{O2}Z_{B2}}\big)\frac{1+Z_{O2}}{Z_{O2}Z_{B2}}\frac{1}{Z_{G2}(1+Z_{O2})}$ \\[.75em]
$\mu_{O}$  & $Z''''_{O1}Z''''_{B2}Z''''_{G2}=\big(Z_{O1}Z_{O2}Z_{G2}\frac{1+Z_{O2}+Z_{O2}Z_{B2}}{1+Z_{G2}+Z_{O2}Z_{G2}}\big)\frac{1+Z_{G2}+Z_{O2}Z_{G2}}{Z_{O2}(1+Z_{B2}+Z_{B2}Z_{G2})}\frac{1+Z_{B2}+Z_{B2}Z_{G2}}{Z_{G2}(1+Z_{O2}+Z_{O2}Z_{B2})}$ \\\bottomrule
\end{tabular}\caption{Dynamics of the Casimir element for the {\color{3ochre}ochre} loop in \Cref{fig:insideout}, up to fourth mutation.}\label{tab:ZO1}
\end{table}

\begin{remark}
The need to ``manually'' update the Casimir elements in this richer geometric picture manifests well-known shortcomings of the classical cluster framework in the presence of self-folding: in our case, mutation formulae \eqref{Xmutation} cannot affect variables attached to isolated vertices. As a direct consequence of this limitation, the quivers involved by the inside-out, displayed side by side in \Cref{fig:amalquivers}, cannot be mapped by standard mutations. In particular, the sequence $\mu_{w_2}$ \eqref{w2mutation} \emph{separately} leaves invariant each of the two quivers. Workarounds to self-folding are known in the literature (e.g., the standard \cite{FST08}) but appear to be incompatible with our amalgamation-driven approach on $\Sigma_{0,4}$.     
Therefore, we simply relax the cluster framework and rely on an elementary Poisson isomorphism to provide the missing identification:
\begin{equation}\label{Poissoniso}
    \begin{matrix}
        Z_{O1} & \longmapsto & y_1y_4y_6, & \quad & Z_{O2} & \longmapsto & y_2, \\
        Z_{B1} & \longmapsto & y_2y_3y_6, & \quad & Z_{B2} & \longmapsto & y_4, \\
        Z_{G1} & \longmapsto & y_2y_4y_5, & \quad & Z_{G2} & \longmapsto & y_6. 
    \end{matrix}
\end{equation}
When there are no self-foldings, variables are named to match the existing literature \cite[(4.4)]{Hikami2019} that depicts the quiver as a octahedron. In particular, notice that the Casimir elements are precisely mapped as prescribed by our geometric recipe, while the remaining variables identify exactly under rotation of $\pi$. 
\end{remark}

\begin{figure}[!t]
    \centering
    \includegraphics[width=.75\textwidth]{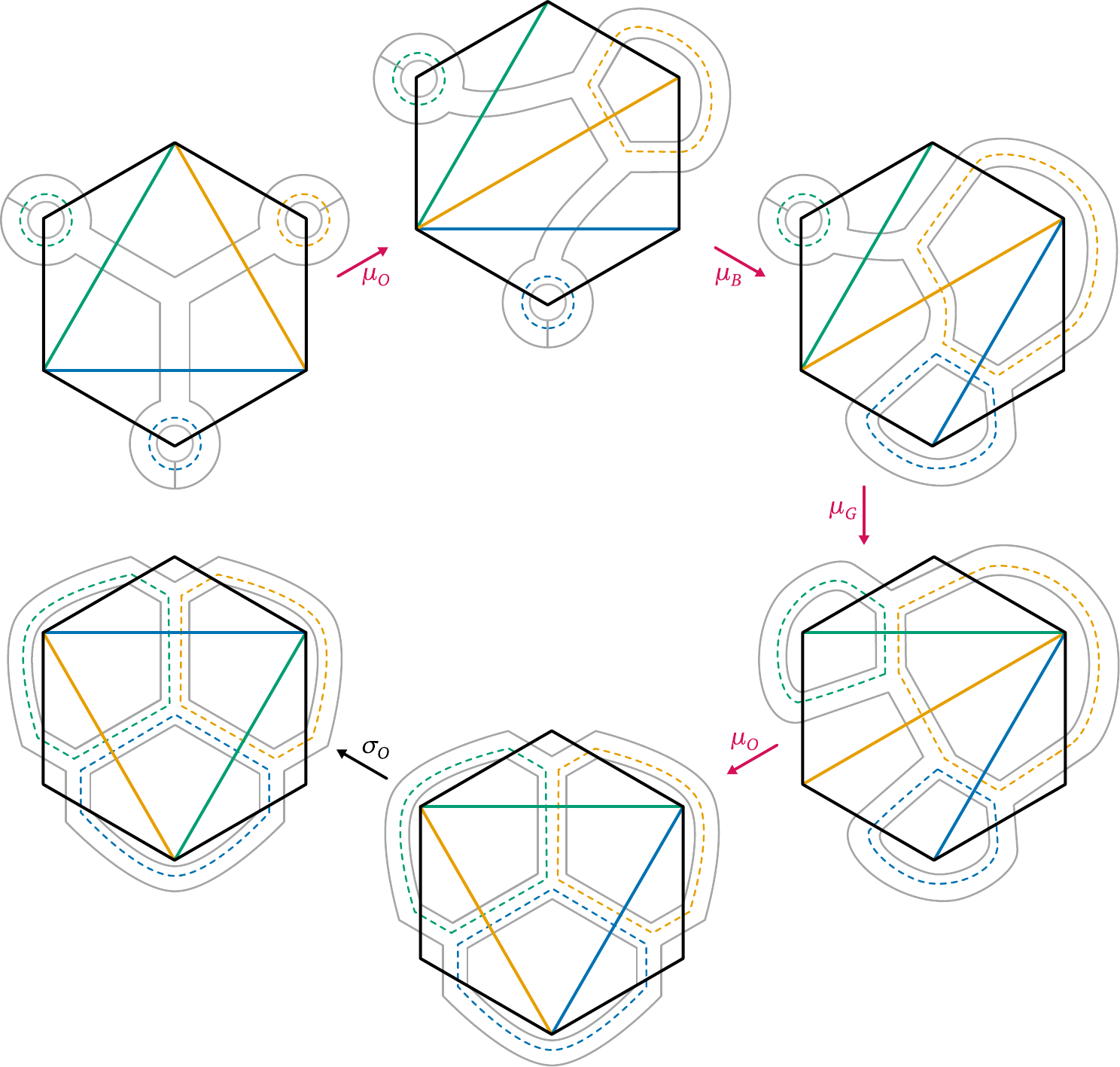}
    \caption{Inside-out of star-shaped fat graphs, dissected into one of its possible factorizations over flips. The dual triangulation is superimposed at each step, with the starting equilateral triangulation as the reference one. Self-glued vertices join at a {\color{mygray}gray} segment and loops are dashed. In cluster terms, chords correspond to subscript $2$ $\pazocal{X}\!$-coordinates, loops to Casimir ones of subscript $1$.}\label{fig:insideout}
\end{figure}

\section{Painlevé squares and the Okamoto cube}\label{sec:cube}

In order to fit the monodromic realization of $w_2$ into the existing literature and solve \Cref{pb2},
we start by giving the $P_{V\!I}$ duality a convolutional formulation.

Harnad's map $\mathfrak{H}^{\!\vee}$ of differential operators relies on generalized Schlesinger equations
\begin{equation}\label{JMMS} 
    \mathrm{d}A_i=-\sum\limits_{j\neq i}[A_i,A_j]\hspace{.1em}\mathrm{d\hspace{.1em}log}(u_i-u_j)-[A_i,\mathrm{d}(u_i\mathfrak{V})+\mathfrak{F}],\quad i=1,\ldots,p,
\end{equation}
for $\mathfrak{V}=\mathrm{diag}(v_1,\ldots,v_n)$, $A_i\in \mathrm{Mat}_n(u,v)$, and a matrix-valued one form $\mathfrak{F}_{ij}$,
which control the isomonodromic deformations of the operator
\begin{equation}
    \frac{\mathrm d}{{\mathrm d}\lambda}-\left(\mathfrak{V}+\sum\limits_{k=1}^p\frac{A_k}{\lambda-u_k}\right).
\end{equation}
For $F,G\in\mathrm{Mat}_{p \times n}(u,v)$ of maximal rank and $U=\mathrm{diag}(u_1,\ldots,u_p)$, the map reads as
\begin{equation}\label{DiffOp}
    \begin{matrix}
        \mathfrak{H}^{\!\vee} & : & \frac{\mathrm d}{{\mathrm d}\lambda}-\left(\mathfrak{V}+G^t(\lambda-U)^{-1}F\right) & \longmapsto & \frac{\mathrm d}{{\mathrm d}z}-\left(U+F(z-\mathfrak{V})^{-1}G^t\right).
    \end{matrix}
\end{equation}
The case of $P_{V\!I}$ selects $n=2,\,p=3,\,\mathfrak{V}=0$ and $U=\mathrm{diag}(0,1,t)$: the dual of \eqref{Sysspace}, with now a double pole at $\infty$, is indeed the $3$-dimensional Birkhoff (Poincaré rank $1$) anticipated system \eqref{irr}.
The duality property ensures that both isomonodromic deformations are governed by the same specialization of \eqref{JMMS}, itself here equivalent to $P_{V\!I}$.

In fact, Harnad duality can be phrased in convolutional terms analogous to those defining the GDAHA functor: map \eqref{DiffOp}, which changes dimension and type of system \eqref{Sysspace} in a single move, is dissected into a dimensional-change step, followed by a type-change step. The former is most efficiently performed as an additive middle convolution $mc$ (cf. \Cref{rmk:freeadd}), while the latter turns out to be a Laplace transform.
In particular, the intermediate $3$-dimensional Fuchsian system has rank $1$ matrix residua, coinciding (up to sign) with row-slices of $V$:
\begin{equation}\label{Fuch3}
    \frac{\mathrm d}{{\mathrm d}\lambda}X=\left(
    \sum_{k=1}^{3}\frac{{\widetilde B}_k}{\lambda-u_k} \right)X \,\Big|\,
    (\widetilde{B}_k)_{ij}=-V_{ij}\delta_{ik},\,\mathrm{eig}(\widetilde{B}_k)=\{0,0,\theta_k\}.
\end{equation}
As such, it can be equivalently written in \={O}kubo normal form \cite[(5.26)]{Filipuk2007}
\begin{equation}
    (\lambda-U)\frac{\mathrm d}{{\mathrm d}\lambda}X=A\,X,\quad A={\widetilde B}_1+{\widetilde B}_2+{\widetilde B}_3=-V,
\end{equation}
and it is then a standard fact for the (inverse) Laplace transform $\mathfrak{L}[X](z):=\oint X(\lambda) e^{z\lambda}\mathrm{d}\lambda=Y(z)$ to turn an \={O}kubo system into the Birkhoff one
$$\frac{\mathrm d}{{\mathrm d}z}Y=\left(U-\frac{A+\One}{z}\right)Y=\left(U+\frac{V-\One}{z}\right)Y$$
of Poincaré rank $1$.
\begin{remark}
    As the middle convolution preserves Schlesinger equations \cite{Filipuk2007}, system \eqref{Fuch3} gives a $3$-dimensional Fuchsian isomonodromy representation. In fact, sequences of preconditioned middle convolutions construct any rank Fuchsian system deformations of which lead to $P_{V\!I}$. Originally, Mazzocco connected this alternative Fuchsian representation with the native $2$-dimensional \eqref{Sysspace} via reduction, provided the use of a special diagonalizing basis \cite[(63)]{Mazzocco2002}. Rethought of as a middle convolution action, this ad hoc step is now promoted to a functorial operation.
\end{remark}

For completeness, we provide the computational details.

Analogously to \Cref{sec:coordinatization}, an explicit parametrization for the middle convolution's input tuple is needed: following \cite{Harnad1994},
\begin{equation}\label{Aform}
{A}_k={1\over2}\left(\begin{array}{cc}
a_k b_k&-a_k^2\\  b_k^2-{\theta_k^2\over a_k^2}&-a_k b_k\\ 
\end{array}\right),
\end{equation}
for complex constants $a_k$ and $b_k$, $k=1,2,3$, such that
\begin{equation}\label{conds}
\sum_{k=1}^3 a_k b_k=-\theta_\infty,\quad  \sum_{k=1}^3 a_k^2=0,\quad
\sum_{k=1}^3 {\theta_k^2\over a_k^2}-b_k^2=0.
\end{equation}
Taking the same additive preconditioner of the Okamoto case \eqref{hatA}, the middle convolution $mc$ delivers the desired dimensional increment.
In particular, the selected change of basis completes the vectors
\begin{equation}
    \left(\frac{a_k^2}{a_kb_k+\theta_k},1\right)\in\ker(\widehat{A}_k), \quad k=1,2,3,
\end{equation}
as
\begin{equation}\label{changevar}
C_a:= \left(  \begin{array}{ccccccc}
                \frac{a_1^2}{a_1b_1+\theta_1} & 0 & 0 & 0 & 0 & 0\\
                1 & 0 & 0 & 1 & 0 & 0\\
                0 & \frac{a_2^2}{a_2b_2+\theta_2} & 0 & 0 & 0 & 0\\
                0 & 1 & 0 & 0 & \frac{a_1}{a_2} & 0\\             
                0 & 0 & \frac{a_3^2}{a_3b_3+\theta_3} & 0 & 0 & 0\\
                0 & 0 & 1 & 0 & 0 & \frac{a_1}{a_3}
            \end{array}
    \right),    
\end{equation}
and the quotient is performed by restricting to the $3\times3$ lowest diagonal blocks of the conjugated $3$-tuple $\left({C_a}^{\!\!\!-1}B_1C_a,{C_a}^{\!\!\!-1}B_2C_a,{C_a}^{\!\!\!-1}B_3C_a\right)$, where $(B_1,B_2,B_3)=c_{0}(\widehat{A}_1,\widehat{A}_2,\widehat{A}_3)$.
The resulting triple $(\widetilde{B}_1,\widetilde{B}_2,\widetilde{B}_3)$, in the anticipated form \eqref{Fuch3} as
\begin{equation}\label{tildeB}
\begin{aligned}
    \widetilde{B}_1 &:= 
                        \begin{pmatrix}
                            \theta_1&\frac{1}{2}\left(a_2 b_1-a_1 b_2+\theta_1{a_2\over a_1}+\theta_2{a_1\over a_2}\right)&\frac{1}{2}\left(a_3 b_1-a_1 b_3-\theta_1{a_3\over a_1}+\theta_3{a_1\over a_3}\right)\\ 
                            0 & 0 & 0\\ 
                            0 & 0 & 0\\
                        \end{pmatrix},\\
    \widetilde{B}_2 &:= 
                        \begin{pmatrix} 
                            0 & 0 & 0\\
                            \frac{1}{2}\left(a_1 b_2-a_2 b_1+\theta_1{a_2\over a_1}+\theta_2{a_1\over a_2}\right)&
                            \theta_2&\frac{1}{2}\left(a_3 b_2-a_2 b_3+\theta_2{a_3\over a_2}+\theta_3{a_2\over a_3}\right)\\
                            0 & 0 & 0
                        \end{pmatrix},\\
    \widetilde{B}_3 &:= 
                        \begin{pmatrix}
                            0 & 0 & 0\\
                            0 & 0 & 0\\
                            \frac{1}{2}\left(a_1 b_3-a_3 b_1+\theta_1{a_3\over a_1}+\theta_3{a_1\over a_3}\right)&
                            \frac{1}{2}\left(a_2 b_3-a_3 b_2+\theta_2{a_3\over a_2}+\theta_3{a_2\over a_3}\right)&\theta_3
                    \end{pmatrix},
\end{aligned}
\end{equation}
sums up to give
\begin{equation}\label{Vmatrix}\resizebox{\textwidth}{!}{$
-V=
\begin{pmatrix}\theta_1&\frac{1}{2}\left(a_2 b_1-a_1 b_2+\theta_1{a_2\over a_1}+\theta_2{a_1\over a_2}\right)&\frac{1}{2}\left(a_3 b_1-a_1 b_3-\theta_1{a_3\over a_1}+\theta_3{a_1\over a_3}\right)\\
\frac{1}{2}\left(a_1 b_2-a_2 b_1+\theta_1{a_2\over a_1}+\theta_2{a_1\over a_2}\right)&
                            \theta_2&\frac{1}{2}\left(a_3 b_2-a_2 b_3+\theta_2{a_3\over a_2}+\theta_3{a_2\over a_3}\right)\\ 
\frac{1}{2}\left(a_1 b_3-a_3 b_1+\theta_1{a_3\over a_1}+\theta_3{a_1\over a_3}\right)&
                            \frac{1}{2}\left(a_2 b_3-a_3 b_2+\theta_2{a_3\over a_2}+\theta_3{a_2\over a_3}\right)&\theta_3\end{pmatrix}.$}
\end{equation}
Such coordinatization of $V$ fixes typos in \cite{Mazzocco2002} and is skew-symmetric for $\theta_1=\theta_2=\theta_3=0$: this specialization of the parameters recovers Dubrovin's operator, whose space of isomonodromic deformations coincide with that of semisimple Frobenius manifolds \cite{Dubrovin1996}.
\begin{remark}\label{rmk:freeadd}
    Much as $\MC$ ignored $\pazocal{L}$, the additive \textbf{parameter-free} $mc$ is defined by excluding $\mathscr{L}$ and thus fails to preserve irreducibility: $\mathrm{ker}(V)=\langle\mathbf{v}\rangle$ implies $\widetilde{B}_k\mathbf{v}=0$, $k=1,2,3$. Nevertheless, the $3\times3$ rank $1$ matrices $\widetilde{B}_k$ are obtained by convolving the $2\times2$ rank $1$ $\widehat{A}_k$s, making this instance of the additive middle convolution preserve the rank $1$ structure. The multiplicative (quantum) analogue of this phenomenon is formalized in \cite[Proposition 6]{DalMartello2024}.
\end{remark}

Now that Harnad duality and $\mathscr{F}_1$ speak the same convolutional language, it is immediate to verify they are compatible in the Riemann-Hilbert sense visualized by the commutative diagram \ref{fig:square}, which is nothing but a refined Painlevé square \eqref{square}.
In particular, commutativity of the left square follows from \Cref{thm:MCcorr} while commutativity of the right one relies on works by Boalch \cite{Boalch2005} and Guzzetti \cite{Guzzetti2016}.

Computation-wise, the diagram's lower row gets explicit on the rescaled coordinatization \eqref{rescOBG} by taking the classical limits in \cite[§5]{DalMartello2024}. In particular \cite[(97)]{DalMartello2024}, the pair of Stokes matrices identifies in $\mathrm{Mat}(\mathbb{C}_\iota[\pazocal{X}_\triangledown])$ with the quadruple of monodromy ones via the coordinatization
\begin{equation}\label{Spmformulae}
\begin{aligned}
    &S_1=U|_{q=1}=\begin{pmatrix}
    1 & -1-Z_{B2}^{-1}-\frac{e^{2\pi i \theta_1}}{Z_{O2}Z_{B2}} & 1+Z_{G2}^{-1}+Z_{B2}^{-1}Z_{G2}^{-1}+\frac{e^{2\pi i \theta_2}}{Z_{B2}Z_{G2}}\Big(1+Z_{B2}^{-1}+\frac{e^{2\pi i \theta_1}}{Z_{O2}Z_{B2}}\Big)\\ 
                    0 & 1 & -1-Z_{G2}^{-1}-\frac{e^{2\pi i \theta_2}}{Z_{B2}Z_{G2}}\\ 
                    0 & 0 & 1
\end{pmatrix},\\
    &S_2=L|_{q=1}=
                \begin{pmatrix}
                    e^{2\pi i \theta_1} & 0 & 0\\
                    e^{2\pi i \theta_2}(1+Z_{O2})+Z_{O2}Z_{B2}  & e^{2\pi i \theta_2} & 0\\
                    Z_{B2}(e^{2\pi i \theta_3}+Z_{G2})+Z_{O2}Z_{B2}Z_{G2}  & e^{2\pi i \theta_3}(1+Z_{B2})+Z_{B2}Z_{G2} & e^{2\pi i \theta_3}
                \end{pmatrix}.
\end{aligned}
\end{equation}
Notice the additive and multiplicative preconditioners underlying the leftmost column of the diagram in \Cref{fig:square} are precisely related by the Riemann-Hilbert correspondence in form \eqref{RHprecond}.
\begin{figure}
    \centering
    \begin{equation*}
        \begin{tikzcd}[row sep=2em, column sep=5em]
    \frac{\mathrm{d}}{\mathrm{d}\lambda}\Phi=\left(\sum_{k=1}^{3}\frac{{\widehat{A}}_k}{\lambda-u_k} \right)\Phi 
        \arrow[d,swap,"\mathrm{RH}"]
      \arrow[r,"mc"]
      \arrow[rr,bend left=15,dashed,"\mathfrak{H}^{\!\vee}"']
    & \frac{\mathrm{d}}{\mathrm{d}\lambda} X=\left(
        \sum_{k=1}^{3}\frac{{\widetilde{B}}_k}{\lambda-u_k} \right)X  \arrow[r,"\sum\widetilde{B}_k=\,-V"]
      \arrow[d,swap,"\mathrm{RH}"]
    & \frac{\mathrm{d}}{\mathrm{d}z}Y=\left(U+\frac{V-\One}{z}\right)Y
      \arrow[d,swap,"\mathrm{RH}"]
    \\
     \arrow[rr,bend right=15,dashed,"\mathscr{F}_1"]
     (\widehat{M}_1,\widehat{M}_2,\widehat{M}_3) \arrow[r,"\MC"] & (R_1,R_2,R_3) \arrow[r,"\prod R_k=\,M_0^{-1}"] & (S_1,S_2)
    \end{tikzcd}
    \end{equation*}
    \caption{Through their convolutional formulations, Harnad duality and the classical GDAHA functor come together to form this commutative diagram, whose rows identifies Fuchsian and Birkhoff formulations of $P_{V\!I}$ in both differential and monodromic frameworks. In particular, operations complementing the middle convolutions, namely the Laplace transform and the Killing factorization (additive and multiplicative case, respectively), are written just as maps of matrix data.}\label{fig:square}
\end{figure}
\begin{remark}
Formulae \eqref{Spmformulae} are extracted by the Killing factorization of a selected representative triple of pseudo-reflections $\bar{\mathbf{R}}=(\bar{R}_1,\bar{R}_2,\bar{R}_3)$ \cite[(96)]{DalMartello2024}. The latter can be obtained straight from convolution using the unique change of basis matrix
\begin{equation}\resizebox{.65\textwidth}{!}{$
    \begin{pmatrix}
                e^{2\pi i \theta_1}Z_{O2}^{-1} & 0 & 0 & 1 & 0 & 0\\
                1 & 0 & 0 & 0 & 0 & 0\\
                0 & -1-e^{-2\pi i \theta_2}Z_{B2}^{-1} & 0 & 0 & -e^{-2\pi i \theta_2} & 0\\
                0 & 1 & 0 & 0 & 0 & 0\\             
                0 & 0 & -1 & 0 & 0 & 0\\
                0 & 0 & 1+e^{2\pi i \theta_3}Z_{G2}^{-1} & 0 & 0 & -Z_{B2}^{-1}Z_{G2}^{-1}
    \end{pmatrix},$}
\end{equation}
whose first three columns are in common with \eqref{changeofvarOka}.
\end{remark}
With full theoretical details provided for the Painlevé square, we are ready to plug $w_2$ into the unfolding diagrammatic picture. 

As anticipated in the Introduction, square \eqref{square} should be understood as a definition of $P_{V\!I}$ at values $(\theta_1,\theta_2,\theta_3,\theta_\infty)$ in terms of its Fuchsian and Birkhoff representations, jointly with monodromy and Stokes data.
Analogously, there is a square corresponding to $P_{V\!I}\left(w_2(\theta)\right)$ at values \eqref{thetachange}, whose Fuchsian side is obtained from that of \eqref{square} via the pair of middle convolutions in diagram \eqref{pb1sol}.

A final piece of literature is needed to unravel the pair of maps for the Birkhoff side.
Change \eqref{thetachange} admits an analogous realization for irregular systems \eqref{irr}, formalized by Mazzocco \cite{Mazzocco2004} as the elementary gauge transformation\footnote{The parameter $\theta_\infty$ appears here with opposite sign with respect to the one in \cite{Mazzocco2004}: this is due to the different, but equivalent, choice of $\vartheta$ between the nonzero eigenvalues of $V$.}
\begin{equation}\label{w2gauge}
    Y \longmapsto \widetilde{Y}:=z^{-\vartheta}Y,
\end{equation}
recalling that $\vartheta=-\,\frac{\theta_1+\theta_2+\theta_3+\theta_\infty}{2}\in\mathrm{eig}(V)$.
It is easily checked the system undergoes the change $V \mapsto V-\vartheta\One$, so that the spectrum of $V$ gets shifted to $\{0,-\vartheta,\theta_\infty\}$.
Thanks to the convolutional-based framework we established, \eqref{w2gauge} easily follows by combining Filipuk's $w_2$ with duality. Indeed, the triple of $3\times3$ rank $1$ matrices resulting from $\widetilde{\mathbf{A}}:=mc_{\vartheta}(\widehat{\mathbf{A}})$, c.f. \cite[(5.16)]{Filipuk2007}, must have spectra $\{0,0,\theta_k+\vartheta\}$, and thus yields the Birkhoff analogue of $\mathrm{Fu}_{\widetilde{\mathbf{A}}}$ as precisely the system
$$\frac{\mathrm d}{{\mathrm d}z}\widetilde{Y}=\left(U+\frac{(V-\vartheta\One)-\One}{z}\right)\widetilde{Y}.$$
Notice that this operation on Birkhoff systems too admits a convolutional formulation: a constant shift on $V$ is nothing but the addition $ad_{(0,-\vartheta)}$ in the generalized sense of \Cref{Takemura}, which translates to an elementary scaling on Stokes data.  
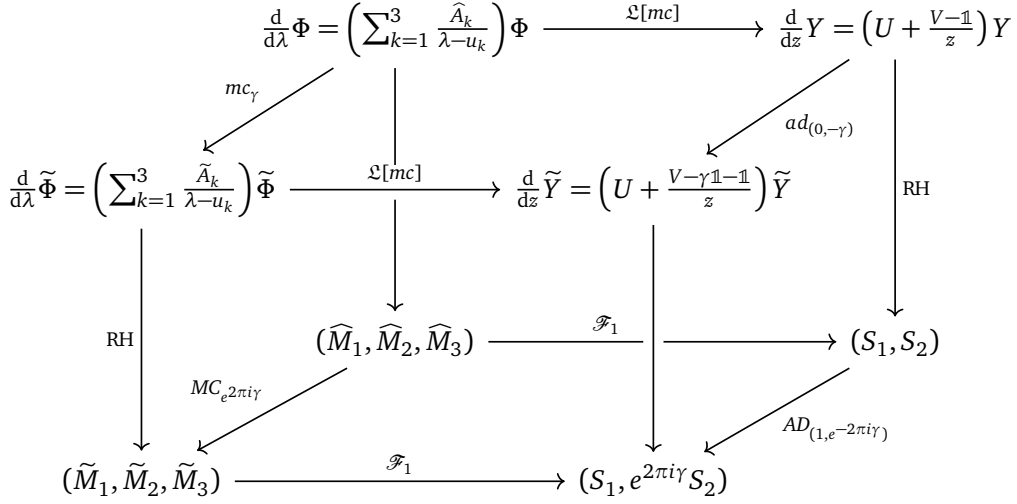
\begin{figure}
    \centering
\begin{equation*}
\begin{tikzcd}[row sep=4em, column sep=-1.5em]
& \frac{\mathrm{d}}{\mathrm{d}\lambda}\Phi=\left(\sum_{k=1}^{3}\frac{\widehat{A}_k}{\lambda-u_k} \right)\Phi \arrow[dl,"mc_{\vartheta}"'] \arrow[rr,"\mathfrak{L}\text{[\emph{mc}]}"] \arrow[dd] & & \frac{\mathrm{d}}{\mathrm{d}z}Y=\left(U+\frac{V-\One}{z}\right)Y \arrow[dl,"ad_{(0,-\vartheta)}"] \arrow[dd,"\mathrm{RH}"] \\
\frac{\mathrm{d}}{\mathrm{d}\lambda}\widetilde{\Phi}=\left(\sum_{k=1}^{3}\frac{\widetilde{A}_k}{\lambda-u_k} \right)\widetilde\Phi \arrow[rr, crossing over, "{\setlength{\fboxsep}{0pt}\colorbox{white}{$\phantom{\widehat{M}}\mathfrak{L}\text{[\emph{mc}]}\phantom{\widehat{M}}$}}"] \arrow[dd,"\mathrm{RH}"'] & & \frac{\mathrm{d}}{\mathrm{d}z}\widetilde{Y}=\left(U+\frac{V-\vartheta\One-\One}{z}\right)\widetilde{Y} \\
& (\widehat{M}_1,\widehat{M}_2,\widehat{M}_3) \arrow[dl,"\MC_{e^{2\pi i \vartheta}}"'] \arrow[rr,"\mathscr{F}_1\hspace{4em}"] & & (S_1,S_2) \arrow[dl,"A\!D_{(1,e^{-2\pi i\vartheta})}"] \\
(\widetilde{M}_1,\widetilde{M}_2,\widetilde{M}_3) \arrow[rr,"\mathscr{F}_1"] & & (S_1,e^{2\pi i\vartheta}S_2) \arrow[from=uu, crossing over]
\end{tikzcd}
\end{equation*} 
\caption{The cube-shaped diagram combining all four convolutional realizations of Okamoto's symmetry. $\mathscr{F}_1$ denotes the classical GDAHA functor, while $\mathfrak{L}[mc]$ indicates the Laplace transformation of the parameter-free additive middle convolution. Faces of this Okamoto cube involve both the Painlevé square (front and back) and diagram \eqref{pb1sol} (left).}\label{fig:cube}
\end{figure}

Combining all four operations in the commutative cube of \Cref{fig:cube}, Okamoto's symmetry acquires a complete diagrammatic realization as the quadruple of convolutional arrows mapping the Painlevé square \eqref{square} of $P_{V\!I}(\theta)$ to the Painlevé square of $P_{V\!I}\left(w_2(\theta)\right)$.
This final all-encompassing diagram, which we name the \textbf{Okamoto cube}, provides a beautiful solution to \Cref{pb2} and thus concludes the paper.

\appendix

\section{Cluster ensembles}\label{app:X&A}

This appendix gives a minimal primer to the theory of cluster ensembles.
We will closely follow \cite{Fock2009}, simplifying notions to a quiver-theoretic setting.

Let $\pazocal{Q}=(Q_0,Q_1,h,t)$ be a quiver: for any arrow $a\in Q_1$, the vertices $h(a),t(a)\in Q_0$ give its head and tail. We restrict to loop-free quivers with no $2$-cycles, i.e., $h(a)\neq t(a)$ and  no opposite arrows connect the same pair of vertices.
Then, the skew-symmetric matrix $(\epsilon_{ij})$, for $\epsilon_{ij}\in\mathbb{Z}$ counting the arrows from vertex $i$ to vertex $j$, is the exchange matrix of $\pazocal{Q}$.

Two algebraic tori $(\mathbb{G}_m)^{|\pazocal{Q}_0|}$ can be attached to any such quiver: the cluster $\pazocal{X}$-torus $\pazocal{X}_\pazocal{Q}$ and the cluster $\pazocal{A}$-torus $\pazocal{A}_\pazocal{Q}$. The former is Poisson, while the latter is endowed with a degenerate closed logarithmic $2$-form $\Omega$.
Over standard coordinates on the tori's factors,
\begin{equation}
    \{X_i,X_j\}=\epsilon_{ij}X_iX_j
\end{equation}
and
\begin{equation}
    \Omega=\epsilon_{ij}\mathrm{d\hspace{.1em}log}\,A_i\wedge\mathrm{d\hspace{.1em}log}\,A_j
\end{equation}
for $1\leq i,j\leq |\pazocal{Q}_0|$.

Both sets $\{X_k\}$ and $\{A_k\}$ are clusters that glue together via mutations into so-called $\pazocal{X}$- and $\pazocal{A}$-varieties, respectively. 
In quiver terms, mutating at vertex $k$ is a $3$-step recipe:
\begin{enumerate}
    \item For each oriented two-arrow path $i\rightarrow k \rightarrow j$, add a new arrow $i \rightarrow j$;
    \item Flip all arrows incident with $k$;
    \item Remove all pairwise disjoint 2-cycles.
\end{enumerate}
The resulting quiver isomorphism $\mu_k:\pazocal{Q}\longmapsto\pazocal{Q}'$ is involutive, namely $\mu^2_k=\mathrm{id}$, and induces a rational map on each torus: on cluster coordinates,
\begin{equation}
    \mu_k^*X'_i=\begin{cases}\label{Xmutation}
        X_k^{-1} \hfill i=k, \\
        X_i\left(1+X_k^{-\mathrm{sgn}(\epsilon_{ik})}\right)^{-\epsilon_{ik}}\hspace{5.2em} i \neq k,
    \end{cases}
\end{equation}
and
\begin{equation}
    \mu_k^*A'_i=\begin{cases}\label{Amutation}
        A_k^{-1}\left(\prod_{j|\epsilon_{kj}>0}A_j^{\epsilon_{kj}}+\prod_{j|\epsilon_{kj}<0}A_j^{-\epsilon_{kj}}\right) \hspace{2em} i=k, \\
        A_i \hfill i \neq k.
    \end{cases}
\end{equation}
In particular, when the set $\{j|\epsilon_{kj}>0\}$ (or $\{j|\epsilon_{kj}<0\}$) is empty, the assigned monomial in \eqref{Amutation} is $1$. 

These mutations respect the Poisson structure and the form $\Omega$, respectively.
Moreover, thanks to the Laurent phenomenon for $\pazocal{A}$-tori, chains of mutation formulae \eqref{Amutation} always result in Laurent polynomials on cluster $\pazocal{A}$-variables.
\begin{definition}
    The pair of spaces $(\pazocal{X}_\pazocal{Q},\pazocal{A}_\pazocal{Q})$ is called a \textbf{cluster ensemble}.
\end{definition}
Since any quiver isomorphism $\sigma=(\sigma_0,\sigma_1)$ preserves ensembles, with
\begin{equation}
    \sigma^* X'_{\sigma_0(i)}=X_i, \quad \sigma^* A'_{\sigma_0(i)}=A_i,
\end{equation} 
mutations are often considered in composition with quiver isomorphisms to form \textbf{cluster transformations}.

We conclude with a key feature of the ensemble structure: the homomorphism
\begin{equation}\label{p}
    p: \pazocal{A_{\pazocal{Q}}}\rightarrow\pazocal{X_{\pazocal{Q}}},
\end{equation}
commuting with mutations and given in every cluster chart by
\begin{equation}
    p^*X_i=\prod\nolimits_jA_j^{\epsilon_{ij}}.
\end{equation}
Fibers of $p$ are the leaves of $\Omega$'s null-foliation, and the subtorus $\pazocal{U}_{\pazocal{Q}}:=p(\pazocal{A}_{\pazocal{Q}})$ is a symplectic leaf of $\pazocal{X}_{\pazocal{Q}}$'s Poisson structure. In particular, the symplectic structure induced on $\pazocal{U}_{\pazocal{Q}}$ by $\Omega$ matches the one given by restricting the Poisson structure on $\pazocal{X}_{\pazocal{Q}}$.

{\footnotesize
\bibliographystyle{abbrv}
\bibliography{references}
}

\end{document}